%% file: Initial_BeamAssociation_mmWave_v4_T9.tex
\documentclass[onecolumn,draftclsnofoot,12pt]{IEEEtran}
\input{input.tex}

\makeatletter
\renewcommand{\maketag@@@}[1]{\hbox{\m@th\normalsize\normalfont#1}}%
\makeatother

\newcommand{\sref}[1]{{Section}~\ref{#1}}
\usepackage{epstopdf}
\usepackage{enumerate}
\usepackage{algorithmicx}
\usepackage{algorithm}
\usepackage{amsmath}
\usepackage[noend]{algpseudocode}
\usepackage{float}
\usepackage{color}
\usepackage{makeidx}
\usepackage{bbm}
\usepackage{subfigure}

\def \Nbs {N_\mathrm{BS}}
\def \Nms {N_\mathrm{MS}}
\def \BS {\mathrm{BS}}
\def \MS {\mathrm{MS}}
\allowdisplaybreaks
	
\begin{document}
\title{Initial Beam Association in Millimeter Wave Cellular Systems: Analysis and Design Insights}
\author{{\fontsize{13pt}{13pt}\selectfont Ahmed Alkhateeb, Young-Han Nam, Md Saifur Rahman, Jianzhong (Charlie) Zhang, and Robert W. Heath Jr.}\thanks{Ahmed Alkhateeb and Robert W. Heath Jr. are with The University of Texas at Austin (Email: aalkhateeb, rheath@utexas.edu). Young-Han Nam, Md Saifur Rahman, Jianzhong (Charlie) Zhang are with Samsung Research America-Dallas (Email: younghan.n, md.rahman, jianzhong.z@samsung.com).} \thanks{This work was done in part while the first author was with Samsung Research America-Dallas. The authors at The University of Texas at Austin are supported in part by the National Science Foundation under Grant No. 1319556. }}
\maketitle

\begin{abstract}

Enabling the high data rates of millimeter wave (mmWave) cellular systems requires deploying large antenna arrays at both the basestations and mobile users. The beamforming weights of these large arrays need to be tuned to guarantee sufficient beamforming gains. Prior work on coverage and rate of mmWave cellular networks focused mainly on the case when base stations and mobile users beamfomring vectors are perfectly designed for maximum beamforming gains. Designing beamforming/combining vectors, though, requires training which may impact both the SINR coverage and rate of mmWave cellular systems. This paper characterizes and evaluates the performance of mmWave cellular networks while accounting for the beam training/association overhead. First, a model for the initial beam association is developed based on beam sweeping and downlink control pilot reuse. To incorporate the impact of beam training into system performance, a new metric, called the effective reliable rate, is defined and adopted. Using stochastic geometry, the effective reliable rate of mmWave cellular networks is derived for two special cases: with near-orthogonal control pilots and with full pilot reuse.  Analytical and simulation results provide insights into the answers of three important questions: (i) What is the impact of beam association on mmWave network performance? (ii) Should orthogonal or reused control pilots be employed in the initial beam association phase? (iii) Should exhaustive or hierarchical search be adopted for the beam training phase?  The results show that unless the employed beams are very wide or the system coherence block length is very small, exhaustive search with full pilot reuse is nearly as good as perfect beam alignment.

\end{abstract}

\section{Introduction} \label{sec:intro}

Wireless communication via millimeter wave (mmWave) frequencies is a key component of future cellular systems \cite{Rappaport2014,Pi2011,Boccardi2014,Rangan2014,Bai2014}. mmWave deployments will use beamforming with large antenna arrays by both the base stations (BSs) and mobile stations (MSs) to ensure sufficient received signal power \cite{Rappaport2013a,Pi2011,Roh2014}. These antenna arrays, though, need to be adaptively steered achieve good beamforming gains. One approach is beam training which aims at selecting the best beamforming/combining vectors from a pre-defined codebook without explicit estimation of the channel \cite{Wang2009,Alkhateeb2014,Hur2013}. The success of beam training is affected by out-of-cell interference, which may lead to misalignment errors and reduce the overall network performance. Therefore, it is important to evaluate the reliable performance of mmWave cellular systems with beam training and association overhead.

\subsection{Prior Work}
Evaluating the coverage and rate performance of mmWave cellular networks has been considered in  \cite{Akdeniz2014,Bai2015,Singh2015,Kulkarni2014,DiRenzo2015}. In \cite{Akdeniz2014}, a statistical channel model was developed to incorporate key findings from  mmWave channel measurements \cite{Rappaport2013a}. This model was then used to assess the performance of mmWave micro and pico-cellular systems via simulations. For more analytical insights, \cite{Bai2015,Singh2015,Kulkarni2014} leveraged stochastic geometry tools to derive the rate and SINR coverage of mmWave cellular networks, incorporating key system features like building blockage and directional antenna arrays. The work in \cite{Bai2015,Singh2015,Kulkarni2014}, however, assumed that the BSs and MSs have perfect beam alignment. Extending these results, \cite{DiRenzo2015,Thornburg2015a} evaluated the performance of  mmWave cellular and ad hoc networks assuming that the transmitter and receiver beams are correctly aligned but with a small random pointing error around the optimal direction. The beam orientation error model adopted in \cite{DiRenzo2015,Thornburg2015a} can be used, for example, to model the pointing errors due to wind \cite{Hur2013}, but may not be suitable for modeling the misalignment errors due to codebook-based initial beam training. The reason is that an error in beam training may cause the BS/MS beams to point in a completely different direction, and not to just have a small error around the optimal direction. Further, the distribution of the random pointing error in \cite{DiRenzo2015,Thornburg2015a} was not a function of the beam training and the out-of-cell interference.
  
Beam training and misalignment errors were investigated in \cite{Wang2009,Alkhateeb2014,Hur2013,BaratiNt.2015,Desai2014}, at the link level with no out-of-cell interference. In the simplest form of beam training, both the BS and MS have predefined codebooks of possible beamforming/combining vectors, and perform beam sweeping  (exhaustive search) by examining all the BS/MS beam pairs. The beam pair that maximizes some performance metric, e.g., the overall beamforming gain, is then selected to be used for data transmission. An exhaustive search is accompanied by high training overhead, especially when narrow beams are employed. To reduce the training overhead, \cite{Wang2009,Alkhateeb2014,Hur2013} designed special codebooks to support hierarchical search. In hierarchical search, the BS and MS first search for the best beam pair over a codebook of wide beam patterns. Then, a second stage of beam refinement is followed where codebooks of narrower beams are used. Mis-alignment error probabilities under different mmWave channel models were derived in \cite{Hur2013,Alkhateeb2014}. Further enhancements on the initial beam training design for mmWave systems were proposed \cite{BaratiNt.2015,Desai2014}, considering more sophisticated system and channel simulations. The work in  \cite{Wang2009,Alkhateeb2014,Hur2013,BaratiNt.2015,Desai2014} did not consider the impact of interference on the beam training process, and how this will affect mmWave cellular network performance. 
\subsection{Contribution}
In this paper, we evaluate the performance of mmWave cellular networks incorporating the impact of initial beam training/ association, and draw some insights into the design of mmWave cellular systems. The contributions of this paper are summarized as follows.
\begin{itemize}
	\item{Developing a  model for analyzing mmWave cellular systems with initial beam training/ association overhead. The model divides the system operation into two phases. In the initial beam association phase, beam sweeping with downlink control pilot reuse is applied by both the BS and MS to find their data transmission beams. In the data transmission phase, the resulting beam pair from the initial beam association phase is employed to send the data from the BS to the mobile user.}
	\item{Analyzing the performance of mmWave cellular networks in terms of a new performance metric called effective reliable rate. This metric captures the resources overhead and SINR penalty of the beam training/association process and accounts for the constraint on the minimum SINR for reliable data transmission. The effective reliable rate is derived for two special cases: full control pilot reuse and near-orthogonal pilots.}
	\item{Leveraging the analysis and simulations results to provide insights into the answers of key questions. (i) What is the impact of beam association on the reliable mmWave network performance? (ii) Should orthogonal or reused control pilots be employed in the initial beam association phase? (iii) Should exhaustive or hierarchical search be adopted for the beam training in mmWave cellular networks? }
\end{itemize}
The results show that mmWave cellular networks can perform nearly as well as with perfect beam alignment, if the initial beam association phase is properly designed. For the control pilot reuse, the results suggest that unless the  beams are very wide, full pilot reuse can provide better effective rates. Simulation results indicate that hierarchical search can have an advantage over exhaustive search only when the system coherence block length is very small. Several other analysis and design insights are also drawn for beam  association in mmWave cellular networks. 


\section{System Model} \label{sec:NetModel}

In this section, we introduce the network, channel, and antenna models considered for evaluating the mmWave network performance in this paper. The parameters defined in this section are summarized in Table \ref{tab:FB}.

\subsection{Network Model}
\textbf{Spatial locations:}
We consider a large-scale mmWave cellular network where the BSs are assumed to be distributed uniformly in $\mathbb{R}^2$ forming a homogeneous Poisson point process (PPP) $\Phi$ with density $\lambda$. In general, BSs can be located inside or outside the buildings. In this paper, we focus on the performance of mmWave cellular networks with only outdoor BSs. Assuming that the indoor and outdoor BSs form two independent PPP's and invoking the thinning theorem \cite{Haenggi2012}, $\Phi\left(\lambda\right)$ can be considered as the PPP of the outdoor BSs \cite{Bai2015}. The users are also assumed to be outdoor and distributed as a stationary point process independently from the BSs. Following the standard stochastic geometry analysis approach in \cite{Andrews2011,Baccelli2009}, a typical user is assumed to be located at the origin, and performance metrics like the downlink SINR and rate distributions are calculated at this user. By the stationarity and independence of the user process, these performance metrics are the same as the network average ones \cite{Baccelli2009}.

\textbf{Blockage model and LOS/NLOS links:}
Recent measurements show that mmWave signals are sensitive to blockage, leading to the need for explicit models for line-of-sight (LOS) and non-line-of-sight (NLOS) path-loss characteristics \cite{Rappaport2013a,MacCartney2013}. We assume the buildings act as propagation blockages for the considered outdoor-outdoor communication between users and BSs. Based on that, the link between the typical user and each BS is considered either LOS or NLOS depending on whether or not one (or more) building blockage intersects the direct link between them. To incorporate that LOS/NLOS difference into our system model, we adopt the blockage model in \cite{Bai2014b,Bai2015}, where the link between the typical user and a BS at distance $r$ is determined to be LOS or NLOS according to a Bernoulli random variable with a LOS probability $p(r)$. In this paper, we consider the LOS probability function $p(r)=e^{-r/\mu}$, where $\mu$ is the LOS range constant that depends on the geometry and density of the buildings \cite{Bai2014b,Bai2015}. 

\textbf{Cell association:}
The typical user is assumed to be associated with the BS that has the smallest path-loss. The cell association is done on a larger time scale compared with the beam association. This means that each user knows the ID of the serving BS in advance of the beam association process. This assumption is also motivated by the recent results in \cite{BaratiNt.2015}, which shows that mmWave cell association can be done efficiently and with low-complexity using omni or quasi-omni beams. Cell association can also be done via lower frequency signaling \cite{Ishii2012} 
\subsection{Channel Model} \label{subsec:Channel_Model}

\begin{table}[t!]
	\caption{System Model Parameters}
	\begin{center}
		\begin{tabular}{ | c || c | }
			\hline
			\textbf{Notation} & \textbf{Description} \\ \hline
			$\Phi\left(\lambda\right)$ & \text{BSs PPP with density $\lambda$} \\ \hline
			$\mu$ & \text{LOS range constant} \\ \hline
			$\rho_\ell$ & \text{distance-dependent path-loss of the $\ell$th BS} \\ \hline
			$\alpha_\mathrm{L}, \alpha_\mathrm{N}$ & \text{path-loss exponent of LOS and NLOS links} \\ \hline
			$C_\mathrm{L}, C_\mathrm{N}$ & \text{intercepts of LOS and NLOS path-loss formulas} \\ \hline
			$\beta_\ell$ & \text{Channel complex coefficient of the $\ell$th BS} \\ \hline			
			$h_\ell$ & \text{Channel power gain of the $\ell$th BS} \\ \hline 
			$\theta_\ell$, $\phi_\ell$ & \text{Channel's angles of arrival and departure} \\ \hline
			$G_\mathrm{BS}$, $G_\mathrm{MS}$ & \text{Main lobe directivity gains of the BS and MS} \\ \hline
			$g_\mathrm{BS}$, $g_\mathrm{MS}$ & \text{Side lobe directivity gains of the BS and MS} \\ \hline
			$\Theta_\mathrm{BS}$, $\Theta_\mathrm{MS}$ & \text{Main lobe half-power beamwidth of the BS and MS} \\ \hline
		\end{tabular}
	\end{center}
	\label{tab:FB}
\end{table}

Let the numbers of BS and MS antennas be $\Nbs$ and $\Nms$, and let the $\Nms \times \Nbs$ matrix $\bH_\ell$ denote the downlink channel from the $\ell$th BS to the typical user. Measurements of outdoor mmWave channels showed that they normally have a limited number of scatterers \cite{Rappaport2013a,Akdeniz2014}. To facilitate analytical tractability, we adopt a geometric channel model similar to \cite{ElAyach2014,Alkhateeb2014},  with a single path between the typical user and each BS. This path can represent a typical LOS channel or an approximation of NLOS with a single scattering cluster. We will show by simulations in \sref{subsec:MP} that the key insights given by the analysis with single-path channels are similar to the case when a few paths exist. Note that the single-path assumption was implicitly adopted by most prior work that use stochastic geometry to study mmWave network performance \cite{Bai2015,Singh2015}. Given this single-path model, the channel $\bH_\ell$ can be written as
\begin{equation}
\bH_\ell = \sqrt{\rho_\ell} \ \beta_\ell \ \ba_\MS\left(\theta_\ell\right) \ba^*_\BS\left(\phi_\ell\right),
\end{equation}
where $\rho_\ell$ represents the distance-dependent path-loss between the typical user and the $\ell$th BS. Let $\alpha_\mathrm{L}$ and $\alpha_\mathrm{N}$ be the path-loss exponents for the LOS and NLOS links, and let $r_\ell$ denote the distance between the typical user and the $\ell$th BS, then $\rho_\ell$ can be defined as $\rho_\ell=C_\mathrm{L} r_\ell^{-\alpha_\mathrm{L}}$ for LOS links and $\rho_\ell=C_\mathrm{N} r_\ell^{-\alpha_\mathrm{N}}$ for the NLOS case, where $C_\mathrm{L}$ and $C_\mathrm{N}$ are the intercepts of the path-loss formulas \cite{Rappaport2013a,MacCartney2013}. The small-scale fading is captured by the complex coefficient $\beta_\ell$. We assume independent Nakagami fading for each link, with different parameters $N_\mathrm{L}$ and $N_\mathrm{N}$ for the LOS and NLOS links. Therefore, $h_\ell=\left|\beta_\ell\right|^2$ is a normalized Gamma random variable. The angles $\theta_\ell$ and $\phi_\ell$ denotes the angles of arrival and departure (AoA/AoD) at the user and the BS. We assume $\theta_\ell$ and $\phi_\ell$ are uniformly distributed in $\left[0, 2 \pi\right]$. Finally, $\ba_\MS\left(\theta_\ell\right)$ and $\ba_\BS\left(\phi_\ell\right)$ are the array response vectors at the MS and BS, respectively, for the angles $\theta_\ell$ and $\phi_\ell$. For simplicity, we assumed  that BSs, MSs, and blockages are distributed in a 2D plane. Extending the proposed framework to 3D channels with azimuth and elevation angles, and with 3D beamforming is a topic for future work.   

\subsection{Antenna Model and Beamforming Gains} \label{subsec:Ant_Model}
\begin{figure}[t]
	\centering
	\subfigure[center][{}]{
		\raisebox{.2\height}{\includegraphics[width=.4\columnwidth]{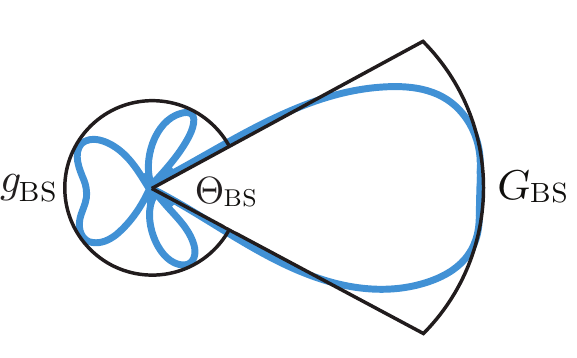}}
		\label{fig:Beam}}
	\subfigure[center][{}]{
		\includegraphics[width=.4\columnwidth]{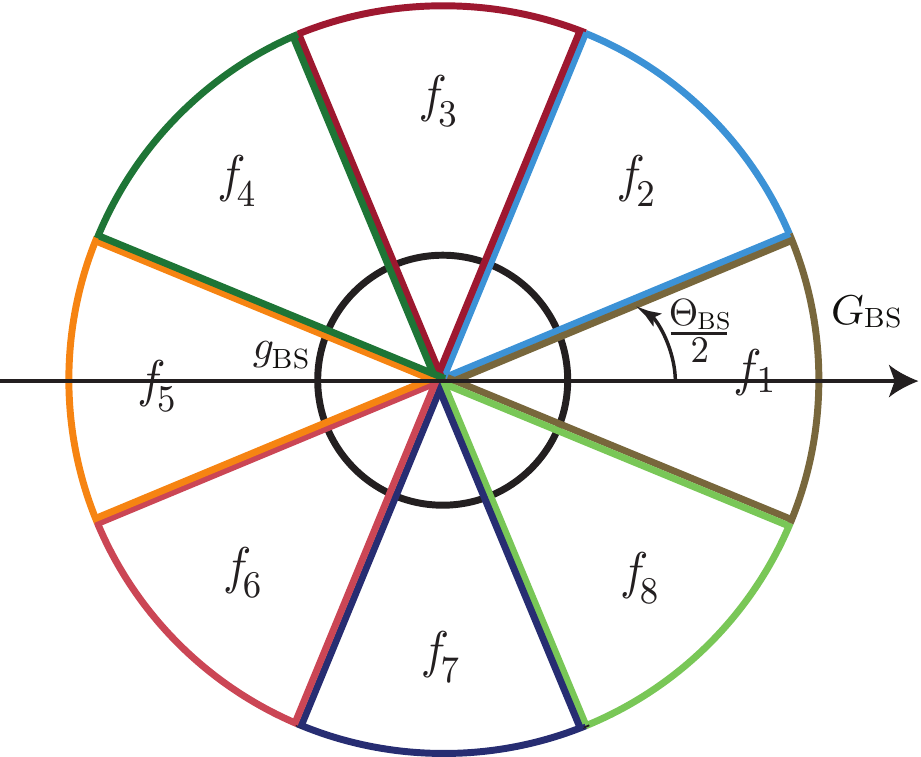}
		\label{fig:Beam_CB}}
	\caption{Figure (a) shows the adopted sectored-pattern antenna model for the BSs with a main-lobe gain $G_\BS$, side-lobe gain $g_\BS$, and main-lobe beamwidth $\Theta_\BS$. Figure (b) illustrates an example codebook of the possible beam patterns for the BS, with a codebook size $N_\BS=8$. The possible beamforming vectors are assumed to have non-overlapping main-lobes.}
\end{figure}

Both the BS and MS employ electronically-steered directional antenna arrays. For simplicity of analysis and exposition, we approximate the actual array beam pattern by a sectorized beam pattern as shown in \figref{fig:Beam}, in which $G_\BS$, $g_\BS$, and $\Theta_\BS$ represent the main-lobe power gain, side-lobe power gain, and main-lobe beamwidth of the BS beam pattern. Similarly, $G_\MS, g_\MS, \Theta_\MS$ define the MS beam pattern. 
 
\textbf{Beam patterns codebook}:
Each BS is assumed to have a codebook of $N_\mathrm{BS}$ possible beamforming vectors $\bff_1, ..., \bff_{N_\BS}$. To simplify the beam training discussion in the following sections, we assume the patterns of these beamforming vectors to have non-overlapping main-lobes designed to cover the full angular range, as shown in \figref{fig:Beam_CB} for $N_\BS=8$. Therefore, $\Theta_\BS$ can be written as $\Theta_\BS=\frac{2 \pi}{N_\BS}$. We define $\mathcal{R}_\BS^{n}$ to be the set of angles of departure covered by the main lobe of $\bff_n, n=1, 2, ..., N_\BS$, i.e., $\cR_\BS^n=\left\{\phi \in [0, 2 \pi]\left| \left|\ba^*_\BS\left(\phi\right) \bff_n \right|^2=G_\BS \right. \right\}$. Note that $\cR_\BS^n$ is invariant to  the reference direction (and the orientation of the antenna array), as long as the AoDs and the angles defining the beamforming vectors have the same reference. Therefore, we model the array orientation of each BS as a uniform random variable in $\left[0, 2\pi\right]$, and let this direction be the BS reference for its AoDs. This direction also describes the boresight direction of the first beamforming vector $\bff_1$.  Similarly, the MS is assumed to have a codebook of $N_\MS$ sectored-pattern combining vectors $\bw_1, ..., \bw_{N_\MS}$. In this paper, we will neglect the MS side lobes, i.e., set $g_\MS=0$. This assumption is important for the tractability of the beam training analysis as will be discussed in \sref{sec:Assoc_Model}. We will, however, show by simulations in \sref{subsec:MP} that the main performance trends given by the adopted sectored pattern is similar to that when both the BSs and MSs employ uniform arrays of reasonable arrays sizes. Constructing multi-resolution beam pattern codebooks with the adopted definition of $\bff_n$ and $\bw_m$ was recently investigated for mmWave systems in \cite{Alkhateeb2014,ElTayeb2015}. The work in \cite{Alkhateeb2014,ElTayeb2015} showed that hybrid analog/digital precoding architectures \cite{Zhang2005a,ElAyach2014,Alkhateeb2013}, which divide the precoding processing between analog and digital domains, can be employed to realize near-ideal sectored beam patterns. 

\textbf{Effective beamforming gain}:
If the channel from the $\ell$th BS and the typical user has an AoD $\phi_\ell$, and this BS uses a beamforming vector $\bff_n$ for signal transmission, then the effective beamforming gain $D_\BS^\ell(n)$, defined as $D_\BS^\ell(n)=\left|\ba^*_\BS\left( \phi_\ell\right) \bff_n\right|^2$, becomes  
\begin{equation}
D_\BS^\ell(n)=
\begin{cases} G_\BS, 
& \phi_\ell \in \cR_\BS^n, \\ g_\BS,
& \phi_\ell \notin \cR_\BS^n. \end{cases}
\end{equation}
The MS beamforming gain $D_\MS^\ell(m)=\left|\ba^*_\MS\left( \theta_\ell\right) \bw_m\right|^2$ can be similarly defined. 

\section{Initial Beam Training and Association Model} \label{sec:Assoc_Model}

To harvest the high beamforming gains promised by large antenna arrays deployed at the BS and MS, they both need to construct their beamforming/combining vectors. Two main approaches can be followed: (i) beam training in which the beamforming/combining vectors are iteratively designed without explicit channel knowledge \cite{Wang2009,Hur2013}, and (ii) beamforming design based on explicit channel estimation \cite{Alkhateeb2014}. Channel estimation can allow more sophisticated MIMO transmission techniques \cite{Alkhateeb2014d,HeathJr2015}. Beam training, though, seems a simpler approach for establishing the initial link after which more enhanced channel estimation algorithms can be applied to enable advanced transmission strategies. Therefore, this paper adopts a beam training model that we describe in this section.  

\begin{figure}[t]
	\centering
	\includegraphics[scale=.9]{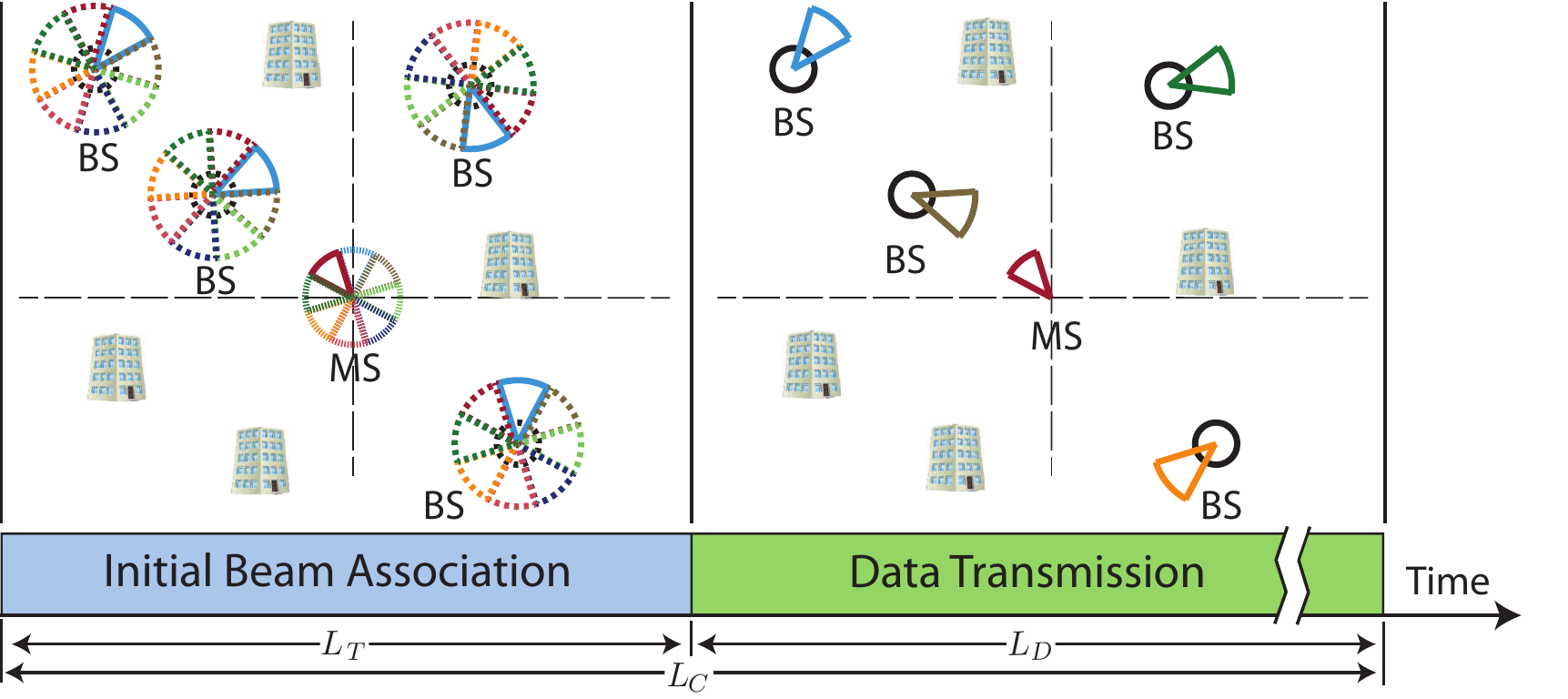}
	\caption{The figure shows that the coherence block length is divided into two phases: (i) a beam training/association phase in which each BS and MS perform beam sweeping to find the best beamforming/combining vectors, and (ii) a data transmission phase which employs the beam pair resulted from the initial beam association phase to send data from the BS to the user.}
	\label{fig:Beam_Assoc}
\end{figure}

First, we assume that the system operation within each time-frequency coherence block will be divided into two phases: an initial beam association phase and a data transmission phase, as depicted in \figref{fig:Beam_Assoc}. In the initial beam association phase, the BSs and MSs train their beamforming/combining vectors to realize the ones that will be employed for sending data during the data transmission phase. Note that systems in practice may not need to repeat the whole beam training process in each coherence block, and may make use of the prior knowledge about the beamforming vectors designed in previous coherence blocks, but we do not exploit that in our model. Hence, if $L_\mathrm{C}$ denotes the number of symbols that can be transmitted within the time-frequency coherence block, and $L_\mathrm{T}, L_\mathrm{D}$ represent the numbers of symbols for the training overhead and transmitted data, then $L_\mathrm{C}=L_\mathrm{T}+L_\mathrm{D}$. In the next subsections, we describe the key elements of the adopted beam training and association model.   
\subsection{Beam Sweeping}
In the beam association phase, we adopt a downlink training model in which each BS and MS perform an exhaustive search over their possible beamforming/combining pairs $\left(\bff_n, \bw_m\right),$ $n=1, ..., N_\BS, m=1, ..., N_\MS$, as shown in \figref{fig:Beam_Assoc}. The MS then selects the beamforming/combining pair $(\bff_{n^\star}, \bw_{m^\star})$ that maximizes a certain metric which will be discussed in \sref{subsec:Beam_Meric}. Further, we assume that the beam association phases of all BSs are aligned in time, and that BSs have synchronous transmission. The requirement of timely-aligned beam association phases can be relaxed in some cases as will be discussed in \sref{subsec:TP_Perf}. Based on this model, if the BSs use a beamforming vector $\bff_n$, and the typical user uses a combining vector $\bw_m$, then the received signal at this user can be written as 
\begin{equation}
y_{(m,n)}=\sum_{\ell: X_\ell \in \Phi(\lambda)}\bw_m^* \bH_\ell \bff_n s_\ell + \bw_m^* \bn_{(m,n)},
\end{equation}
where $\bn_{(m,n)} \sim \mathcal{CN}\left(\boldsymbol{0}, \sigma^2 \bI\right)$ is the noise vector at the receiver, and $s_\ell$ is the transmitted training symbol of the $\ell$th BS, $\bbE\left[s_\ell s_\ell^*\right]=P_T$, with $P_T$ the average transmit power. It is important to note here that as a result of the random array orientation of each BS, the $n$th beamforming vector $\bff_n$ of the different BSs correspond to different beamforming directions when compared to a common reference. Now, given the channel and antenna models in \sref{subsec:Channel_Model}-\sref{subsec:Ant_Model}, the BSs PPP $\Phi(\lambda)$ can be divided into $N_\MS$ independent PPP's $\Phi_m\left(\lambda\right), m=1, ..., N_\MS$, where $\Phi_m\left(\lambda\right)$ consists of the BSs whose AoA's lie within the main-lobe of the $m$th combining vector of the typical user. Formally, $\Phi_m\left(\lambda\right)$ can be defined as $\Phi_m\left(\lambda\right)=\left\{X_\ell \in \Phi(\lambda)\left| \right. \theta_\ell \in \cR_\MS^m \right\}$. It is worth noting here that the independence condition exists because of the single-path channel and zero MS side-lobe assumptions in the system model. While this condition is important for the analytical tractability, we show by simulations in \sref{sec:Results} that the main insights if the paper are applicable to more general settings. Applying the definition of $\Phi_m\left(\lambda\right)$ and employing the effective beamforming gain in \sref{subsec:Ant_Model}, the received signal $y_{(m,n)}$ can be expressed as 
\begin{equation}
	y_{(m,n)}=\sum_{\ell: X_\ell \in \Phi_m(\lambda)} \beta_\ell \left(P_T \rho_\ell G_\MS D_\BS^\ell(n) \right)^{\frac{1}{2}} s_\ell + \bw_m^* \bn_{(m,n)}.
	\label{eq:Rec_Sig}
\end{equation} 
The received signals of the other beamforming/combining pairs are similarly written.  Note that the summation of the first term of \eqref{eq:Rec_Sig} includes the desired signal if  the serving BS is in $\Phi_m(\lambda)$.

\subsection{Downlink Control Pilot Reuse} \label{subsec:Pilots}

In the downlink training, the MS differentiates its serving BS signal from the other BSs signals using the control pilot sequence assigned to its BS. For the MS to detect its serving BS signal free of interference, orthogonal pilot sequences need to be assigned to all interfering BSs. This, however, may require many pilots and is difficult to do in practice. Therefore, we assume that BSs will reuse the downlink control pilots. Further, for analytical tractability, we assume these pilots to be reused randomly over the BSs locations. To incorporate that in our model, we let $\delta_c \in [0, 1]$ denote the fractional control pilot reuse factor, and let each BS reuse the same pilot of the serving BS with an IID probability $\delta_c$. Then, by invoking the thinning theorem, \cite{Haenggi2012}, the BS PPP $\Phi(\lambda)$ is partitioned into two independent non-homogeneous PPP's: the PPP of control-level interfering BSs $\Phi^\mathrm{c}\left(\lambda_\mathrm{c}\right)$, with density $\lambda_c=\delta_c \lambda$, which includes the co-pilot BSs that share the same pilot with the serving BS, and the PPP $\Phi^\mathrm{o}((1-\delta_c) \lambda)$ with the other BSs which employ orthogonal pilots. Note that this PPP partitioning to co-pilot and orthogonal-pilot BSs is applied to each BS PPP $\Phi_m\left(\lambda\right), m= 1, 2, ..., N_\mathrm{MS}$, which is also divided into co-pilot PPP $\Phi_{m}^\mathrm{c}\left(\lambda_\mathrm{c}\right)$, and orthogonal-pilot PPP $\Phi_m^\mathrm{o}\left((1-\delta_\mathrm{c})\lambda\right)$. Recall that the BS PPP $\Phi_m\left(\lambda\right)$ consists of the BSs whose AoA's lie within the main-lobe of the $m$th combining vector of the typical use. Finally, if $\Pi_\mathrm{p}$ denotes the number of orthogonal pilots, then $\Pi_\mathrm{p}=\frac{1}{\delta_c}$.   

\subsection{Beam Pair Decision Criterion} \label{subsec:Beam_Meric}

We adopt the Max SNR approach as a beam-pair decision criterion. This means that the MS calculates the resultant SNR with each beamforming/combining pair, and selects the beam pair with the maximum SNR to be used  in the data transmission phase. Note that the time duration of each beam pair transmission is assumed to be enough for the estimation of the received energy at the MS. Further, the BS pilot sequences, which are repeated for each beam pair, are used to cancel the interference from the BSs with orthogonal pilots. Next, we  describe this Max SNR beam pair selection criteria. Applying the control pilot-reuse model in \sref{subsec:Pilots}, the  received signal in \eqref{eq:Rec_Sig}, with the beamforming/combining pair $\left(\bff_n, \bw_m\right)$, can be written as 
\begin{equation}  
	y_{(m,n)}=\hspace{-7pt}\sum_{\ell: X_\ell \in \Phi_m^\mathrm{c}(\lambda_\mathrm{c})} \hspace{-3pt} \beta_\ell \left(P_T \rho_\ell G_\MS D_\BS^\ell(n) \right)^{\frac{1}{2}} s_\ell + \hspace{-7pt}\sum_{\ell: X_\ell \in \Phi_m^\mathrm{o}((1-\delta_\mathrm{c})\lambda)} \hspace{-3pt} \beta_\ell \left(P_T \rho_\ell G_\MS D_\BS^\ell(n) \right)^{\frac{1}{2}} s_\ell + \bw_m^* \bn_{(m,n)}.
\end{equation}
where the first term represents the received signal from the serving BS plus the signal from the co-pilot interfering BSs, and the second term represents the received signal from the BSs with orthogonal pilots. Note that the use of orthogonal pilots during the control phase causes the resultant SNR in this beam association phase to be different than the data-transmission SNR, as the received power from the co-pilot interfering BSs will be seen as a desired signal power at the MS. Let $\mathsf{SNR}_{m,n}^\mathrm{c}$ denote the control-level SNR due to beamforming/combining pairs $\left(\bff_n, \bw_m\right)$:
\begin{equation} 
\mathsf{SNR}_{m,n}^\mathrm{c}=\frac{\sum_{\ell: X_\ell \in \Phi_m^\mathrm{c}(\lambda_\mathrm{c})} P_T h_\ell \rho_\ell G_\MS D_\BS^\ell(n)} {\sigma^2}. \label{eq:c_SINR}
\end{equation}
After sweeping over all the beamforming/combining pairs, the MS will select the best beamforming/combining pair $\left\{\bff_{n^\star}, \bw_{m^\star}\right\}$, that solves 
\begin{equation}
\left\{\bff_{n^\star}, \bw_{m^\star}\right\} = \displaystyle{\operatorname*{\arg max}_{\substack{n=1, ..., N_\BS \\ m=1, ..., N_\MS}}} \mathsf{SNR}_{m,n}^\mathrm{c}. 
\end{equation}
The obtained beamforming/combining beam patterns $\left\{\bff_{n^\star}, \bw_{m^\star}\right\}$, will then be employed by the MS and its serving BS during the data transmission phase.

\section{Coverage and Effective Rate} \label{sec:Analysis}

The objective of this section is to characterize the performance of mmWave cellular networks taking into consideration the impact of the initial beam association phase. In \sref{subsec:Metrics}, we introduce the metrics that we adopt for evaluating the mmWave network performance, before characterizing this performance in \sref{subsec:SINR_Perf}-\sref{subsec:TP_Perf}.

\subsection{Performance Metrics} \label{subsec:Metrics}
The impact of the initial beam association phase is found in two places: (i) resources overhead which represents the time-frequency resources allocated to the beam training/association phase, and (ii) SINR overhead which is the penalty on the SINR coverage due to the possible misalignment error in the beam training phase. To capture these two sources of overhead due to initial beam association, we define and adopt the following performance metrics. The adopted metrics also take into consideration the point that practical communication systems may not transmit reliably with any modulation and coding scheme (MCS) below a certain SINR threshold. 
\begin{itemize}
	\item{\textbf{SINR Coverage:}
		Here, the SINR coverage, ${P}_\mathrm{c}\left(T\right)$, is defined as the probability that the received SINR at the typical user during the data transmission phase is greater than a threshold $T > 0$, when the beamforming/combining vectors resulting from the beam association phase are employed. Note that the single-path channel model and zero MS side lobe assumptions in \sref{subsec:Channel_Model}-\sref{subsec:Ant_Model} imply that the received SINR can be greater than zero if and only if $D_\MS^0 = G_\MS$, where the index $\ell=0$ will be reserved in this paper for the serving/tagged BS. Therefore, we define the following two events. In the optimal beam pair ($\mathsf{OBP}$) event, the selected beamforming/combining pair in the beam association phase is the one that maximizes the BS and MS beamforming gains, i.e., $\mathsf{OBP}=\left\{D_\BS^0=G_\BS, D_\MS^0=G_\MS\right\}$. In the sub-optimal beam pair ($\mathsf{SBP}$) event, the MS beamforming gain is maximized, but not the BS gain, i.e., $\mathsf{SBP}=\left\{D^0_\BS=g_\BS, D^0_\MS=G_\MS\right\}$. Note that those are the only cases with non-zero SINR coverage probability. Therefore, we can expand ${P}_\mathrm{c}\left(T\right)$ as 
		\begin{equation}
		{P}_\mathrm{c}\left(T\right)=\mathbb{P}\left(\mathsf{SINR} > T \left| \mathsf{OBP} \right.\right) \mathbb{P}\left(\mathsf{OBP}\right)+\mathbb{P}\left(\mathsf{SINR} > T \left| \mathsf{SBP} \right.\right)  \mathbb{P}\left(\mathsf{SBP}\right), \label{eq:SINR_Cov}
		\end{equation}
		where $\mathsf{SINR}$ is the data  transmission SINR with the beamforming/combining vectors resulting from the beam training and association phase.
	}
	\item{\textbf{Effective Achievable Rate:}\
	To evaluate the aggregate performance of a mmWave cellular network, we propose an effective rate performance metric, which captures both the resource overhead and the SINR penalty of the beam association phase as well as accounting for the constraint on the minimum SINR. If $\eta$ represents the system resource efficiency---the percentage of the system resources allocated for data transmission---the effective achievable rate is defined as
	\begin{equation}
	R_\mathrm{eff}= \eta \ \bbE \left[  \log_2\left(1+\max\left(\mathsf{SINR}, T_\mathrm{max}\right)\right) \mathbbm{1}_{\mathsf{SINR} \geq T_\mathrm{th}}\right],\label{eq:eff_Rate}
	\end{equation}
	with the SINR thresholds $T_\mathrm{th}$ and $T_\mathrm{max}$ denoting the minimum and maximum SINRs supported by the modulation and coding scheme (MCS). The system resources efficiency, $\eta(\lambda, \Nbs, \Nms)$, is defined as
	\begin{equation} \label{eq:Eff_Def}
	\eta=\left(1-\frac{\Pi_\mathrm{p}(\lambda) \Nbs \Nms}{L_\mathrm{C}}\right)^+.
	\end{equation}
	Finally, given the SINR coverage probability $P_\mathrm{c}\left(T\right)$, the effective average rate can be expressed as \cite{Andrews2011}
	\begin{equation}
		R_\mathrm{eff}= \left(1-\frac{\Pi_\mathrm{p}(\lambda) \Nbs \Nms}{L_\mathrm{C}}\right)^+ \left[\frac{1}{\ln(2)} \int_{T_\mathrm{th}}^{T_\mathrm{max}} \frac{{P}_{\mathrm{c}}\left(y\right)}{y+1} dy +  \log_2\left(1+T_\text{th}\right) P_{c}\left(T_\text{th}\right)\right].
		\label{eq:eff_Rate_F}
	\end{equation}
	}
\end{itemize}

Analyzing the performance of beam-association based mmWave cellular systems is in general non-trivial. This is mainly due to the correlation between the control-level SNR $\left(\text{denoted }  \mathsf{SNR}^\mathrm{c}\right)$ in \eqref{eq:c_SINR} and the data transmission SINR, as they belong to the same BS PPP's and are calculated for the same channels. This makes it difficult to characterize the SINR coverage probability in \eqref{eq:SINR_Cov} for an arbitrary pilot reuse factor. Therefore, we will analyze the mmWave network performance for two important special cases in the following two subsections:  when the control pilot reuse factor is very small and when the pilot reuse factor is one. Investigating those two extreme cases helps us draw conclusions about the impact of beam association on mmWave cellular performance as well as make observations on more general settings. 

\subsection{Performance with Near-Orthogonal Pilots}\label{subsec:SINR_Perf}

As explained in \sref{subsec:Pilots}, the number of orthogonal pilots can be written as $\Pi_\mathrm{p}=\frac{1}{\delta_\mathrm{c}}$. This means that the number of orthogonal pilots can, theoretically, go to infinity when the pilot reuse factor approaches zero. In practice, however, there is a minimum value for the pilot reuse factor after which the co-pilot interference can be neglected with respect to the noise power. Therefore, we will  let  $\lambda_{\mathrm{c},\text{min}}$ denote the control-level interfering BS density at which the co-pilot interference to noise ratio is close to zero with high probability, i.e., $\mathbb{P}\left(\frac{\sum_{\ell: X_\ell \in \Phi_m^\mathrm{c}(\lambda_{\mathrm{c},\text{min}}), \ell \neq 0} P_T h_\ell \rho_\ell G_\MS D_\BS^\ell(n)}{\sigma^2} < \epsilon_1 \right) \geq 1- \epsilon_2$ for very small $\epsilon_1, \epsilon_2$. Further, we introduce the interference radius term $R_\mathrm{I}$ and define it as the circle radius at the minimum control BS density $\lambda_{\mathrm{c},\text{min}}$ that has only one co-pilot interfering BS on average, i.e., $ \pi R_\mathrm{I}^2 \lambda_{\mathrm{c},\text{min}}=1$. We can then write $\delta_{\mathrm{c},\text{min}}=\frac{1}{ \pi \lambda R_\mathrm{I}^2}$, and the number of orthogonal pilots in this case as $\Pi_\mathrm{p}(\lambda)=\pi \lambda R_\mathrm{I}^2$. It is worth noting here that $R_\mathrm{I}$ is a network constant that depends on the control BS density $\lambda_\mathrm{c}=\delta_{c} \lambda$, and not on the actual BS density $\lambda$. Next, we derive the SINR coverage and effective rates metrics for this case. 
\begin{itemize}
	\item{\textbf{SINR Coverage:}
		Given the minimum pilot reuse factor $\delta_{\mathrm{c}, \text{min}}$ with which the co-pilot interference is negligible, the control-level SNR in \eqref{eq:c_SINR} can be approximated as 
		\begin{equation} \label{eq:SNR_C_Near}
		\mathsf{SNR}_{m,n}^\mathrm{c} \approx \frac{P_T h_0 \rho_0 D_\BS^0(n) D_\MS^0(m)} {\sigma^2},
		\end{equation}     
		which implies that the beam association phase will result in the maximum BS and MS beamforming gains $G_\BS, G_\MS$ with high probability, i.e., $\mathbb{P}\left(\text{OBP}\right) \approx 1$, as the impact of the co-pilot interference is omitted in \eqref{eq:SNR_C_Near}. Note that $\mathbb{P}\left(\text{OBP}\right) \approx 1$ also means that $\mathbb{P}\left(\text{SBP}\right) \approx 0$. Therefore, the SINR coverage in \eqref{eq:SINR_Cov} can be approximated as
		\begin{equation}
		P_{\mathrm{c}}^{\text{Orth}}\left(T\right) \approx \mathbb{P}\left(\frac{P_T G_\BS G_\MS h_0  \rho_0 }{  \sum_{\ell: X_\ell \in \Phi_m (\lambda), \ell \neq 0} P_T h_\ell \rho_\ell G_\MS D_\BS^\ell(n) + \sigma^2} > T\right).
		\label{eq:Pc_SINR}
		\end{equation}
		The SINR coverage expression in \eqref{eq:Pc_SINR} is similar to that with perfect beam alignment characterized before in \cite{Bai2015} whose results can therefore be used to determine $P_{\mathrm{c}}^{\text{Orth}}\left(T\right)$.
	}
	\item{\textbf{Effective Achievable Rate:}
        Given the number of required pilots, $\Pi_\mathrm{p}(\lambda)=\pi \lambda R_\mathrm{I}^2$, the system resource efficiency in \eqref{eq:Eff_Def} becomes 
        \begin{equation}
        \eta=\left(1-\frac{2 \pi \lambda R_\mathrm{I}^2 N_\BS N_\MS}{L_\mathrm{C}}\right)^+,
        \end{equation}
        and the effective achievable rate can be expressed as 
        \begin{equation}
        R_\mathrm{eff}= \left( 1-\frac{ \pi \lambda R_\mathrm{I}^2 N_\BS N_\MS}{L_\mathrm{C}} \right)^+ \left[\frac{1}{\ln(2)} \int_{T_\mathrm{th}}^{T_\mathrm{max}} \frac{{P}^\text{Orth}_{\mathrm{c}}\left(y\right)}{y+1} dy +  \log_2\left(1+T_\text{th}\right) P^\text{Orth}_{c}\left(T_\text{th}\right)\right].
        \label{eq:Rate_SINR}
        \end{equation}
        Note that while the SINR penalty of the beam association phase is negligible, the number of required orthogonal pilots can be high resulting in large training overhead and small system resources efficiency.  
	}
\end{itemize}

\subsection{Performance with Full Pilot Reuse} \label{subsec:TP_Perf}

In this subsection, we consider the other extreme case with full control pilot reuse ($\delta_\mathrm{c} = 1$). In this case, the MS will not be able to differentiate between the signals of its serving BS and the other BSs, and will make its beam pair decisions based on the total (signal plus interference) received power. Note that the synchronization of the beam association phases among the different cells is not required in this case, as all the BSs use the same pilot, and similar interference can be caused by the control and data signals. To simplify the analysis, we will assume that the typical user and its serving BS are performing the beam training and association while the other BSs are sending data with random beamforming vectors. Next, we characterize the SINR coverage and effective rate of the mmWave cellular networks with full downlink control pilot reuse. 
\begin{itemize}
	\item{\textbf{SINR Coverage:}
	Reusing the same pilot by all the BSs in the beam association phase can cause non-negligible beam-alignment errors, which affect the SINR coverage. To analyze this, we start by understanding the $\mathsf{OBP}$ event. For the full pilot reuse beam association model, the $\mathsf{OBP}$ event occurs when the total power in the direction of the optimal beam pair---the one that maximizes the BS/MS beamforming gains---is greater than the total power in any other direction at the MS. We can, therefore, write the $\mathsf{OBP}$ event as
		\begin{equation}
		\mathsf{OBP}=\left\{P_T h_0 \rho_0  G_\BS G_\MS + I_{\Phi_{1}} > I_{\Phi_{2}}, ..., I_{\Phi_{N_\MS}}\right\}, \label{eq:OBP_TP}
		\end{equation}
		where $I_{\Phi_{m}} =  \sum_{\ell: X_\ell \in \Phi_m (\lambda), \ell \neq 0} P_T h_\ell \rho_\ell G_\MS D_\BS^\ell(n), m=1, ..., N_\MS$ is the interference power when the $m$th MS combining vector is used. Note that we assumed, without loss of generality, that the MS combining gain to the serving BS is maximized when the combining vector $m=1$ is used, i.e., $D_\MS^0(1)=G_MS$. Next, we derive upper and lower bounds for SINR coverage probability with full control pilot reuse $P_\mathrm{c}^{\text{Reused}}(T)$. 
		\begin{theorem}
			The SINR coverage probability of the mmWave cellular network described in \sref{sec:NetModel}, under the initial beam association model in \sref{sec:Assoc_Model} with full control pilot reuse, is
			\begin{equation}
			P_\mathrm{c}^\text{Reused} (T)  =P^\text{Reused}_\text{c$|$$\mathrm{L}$}\left(T\right) A_\mathrm{L} + P^\text{Reused}_\text{c$|$$\mathrm{N}$}\left(T\right) A_\mathrm{N},
			\end{equation}
			with $P^\text{Reused}_\text{c$|$$\mathrm{L}$}\left(T\right)$ and $P^\text{Reused}_\text{c$|$$\mathrm{N}$}\left(T\right)$ upper bounded by 
			\begin{equation}
			P^\text{Reused}_\text{c$|$$\mathrm{L}$}\left(T\right)\leq \sum_{n=1}^{N}  \left(-1\right)^{n+1}  \dbinom{N}{n} \int_{0}^{\infty}{\int_{0}^{\infty}{e^{-\Upsilon_{n}^1(T,g,r)-{\Upsilon }_{n}^2(T,g,r) } f_{h_0}^\mathrm{L} \left(g\right) } f_{r_0}^\mathrm{L}(r) dg dr}, \label{eq:CondTP1}
			\end{equation}
			\begin{equation}
			P^\text{Reused}_\text{c$|$$\mathrm{N}$}\left(T\right) \leq \sum_{n=1}^{N} \left(-1\right)^{n+1}  \dbinom{N}{n} \int_{0}^{\infty}{\int_{0}^{\infty}{e^{-\Upsilon^3_n(T,g,r)-\Upsilon^4_n(T,g,r)}     f_{h_0}^\mathrm{N} \left(g\right) } f_{r_0}^\mathrm{N}(r) dg dr},
			\end{equation}
			where
			\begin{align}
			& \Upsilon^1_n(T,g,x)= 2 \pi \lambda \sum_{k=1}^2 b_k \int_x^\infty {F\left(N_\mathrm{L},\frac{a n C_\mathrm{L} P_T G_k t^{- \alpha_\mathrm{L}}}{\tau_\mathrm{L}(g, x) N_\mathrm{L}} \right) t p(t) dt}, \label{eq:Up1_U}\\
			& \Upsilon^2_n(T,g,x)= 2 \pi \lambda \sum_{k=1}^2 b_k \int_{\psi_\mathrm{L}(x)}^\infty { F\left(N_\mathrm{N},\frac{a n C_\mathrm{N} P_T G_k t^{- \alpha_\mathrm{N}}}{\tau_\mathrm{L}(g, x) N_\mathrm{N}} \right) t (1-p(t)) dt }, \label{eq:Up2_U}\\
			& \Upsilon^3_n(T,g,x)= 2 \pi \lambda \sum_{k=1}^2 b_k \int_{x}^\infty { F\left(N_\mathrm{N},\frac{a n C_\mathrm{N} P_T G_k t^{- \alpha_\mathrm{N}}}{\tau_\mathrm{N}(g, x) N_\mathrm{N}}\right) t (1-p(t)) dt }, \\			
			& \Upsilon^4_n(T,g,x)= 2 \pi \lambda \sum_{k=1}^2 b_k \int_{\psi_\mathrm{N}(x)}^\infty {F\left(N_\mathrm{L}, \frac{a n C_\mathrm{L} P_T G_k t^{- \alpha_\mathrm{L}}}{\tau_\mathrm{N}(g, x) N_\mathrm{L}}\right) t p(t) dt },		
			\end{align}
			\noindent $F(N,x)=1-(1)/(1+x)^N$, and $A_\mathrm{L}, A_\mathrm{N}$ are the probabilities of the user associated to LOS and NLOS BSs. The constants $b_k=\frac{(N_\BS-2)(k-1)+1}{N_\BS}$, $G_k=G_\BS G_\MS$ for $k=1$ and $G_k=g_\BS G_\MS$ for $k=2$. The functions $f_{h_0}^\mathrm{L}(g)=f_\gamma(g,N_\mathrm{L},1/N_\mathrm{L})$ and $f_{h_0}^\mathrm{N}(g)=f_\gamma(g,N_\mathrm{N},1/N_\mathrm{N})$ are the PDFs of normalized gamma random variable modeling the LOS and NLOS channel gains, with $f_\gamma(g,N,1/N)=\frac{N^N x^{N-1} e^{-N x}}{(N-1)!}$. Finally, the functions $f_{r_0}^\mathrm{L}(x)$ and $f_{r_0}^\mathrm{N}(x)$ are the PDFs of the distance between the MS and its LOS or NLOS serving BS.
			\label{th:thm1}
		\end{theorem}
		\begin{proof}
		See Appendix \ref{app:Full_Pilot_UB}.
		\end{proof}
		
		\begin{theorem}
					The SINR coverage probability of the mmWave cellular network described in \sref{sec:NetModel}, under the initial beam association model in \sref{sec:Assoc_Model} with full control pilot reuse, is
					\begin{equation}
					P_\mathrm{c}^\text{Reused} (T)  =P^\text{Reused}_\text{c$|$$\mathrm{L}$}\left(T\right) A_\mathrm{L} + P^\text{Reused}_\text{c$|$$\mathrm{N}$}\left(T\right) A_\mathrm{N},
					\end{equation}
					with $P^\text{Reused}_\text{c$|$$\mathrm{L}$}\left(T\right)$ and $P^\text{Reused}_\text{c$|$$\mathrm{N}$}\left(T\right)$ lower bounded by
					\begin{align}
					P^\text{Reused}_\text{c$|$$\mathrm{L}$}\left(T\right) & \geq \sum_{n=1}^{N_\mathrm{L}}  \left(-1\right)^{n+1} \dbinom{N_\mathrm{L}}{n}  \int_{0}^{\infty}{    e^{-\frac{n a_\mathrm{L} x^{-\alpha_\mathrm{L}} T \sigma^2}{C_\mathrm{L} G_\BS G_\MS}-\overline{\Upsilon \vphantom{(a)}}^1_n(T,x)-\overline{\Upsilon \vphantom{(a)}}^2_n(T,x)}     f_{r_0}^\mathrm{L}(x) dx}, \\
					& \hspace{20pt} \times \left[\sum_{n=1}^{N_\mathrm{L}}  \left(-1\right)^{n+1}  \dbinom{N_\mathrm{L}}{n} \int_{0}^{\infty}{    e^{\overline{\Upsilon \vphantom{(a)}}^1_n(T,x)-\overline{\Upsilon \vphantom{(a)}}^2_n(T,x)}     f_{r_0}^\mathrm{L}(x) dx}\right]^{\Nms-1}, \nonumber
					\end{align}
					and 
					\begin{align}
					P^\text{Reused}_\text{c$|$$\mathrm{N}$}\left(T\right) & \geq \sum_{n=1}^{N_\mathrm{N}}  \left(-1\right)^{n+1} \dbinom{N_\mathrm{N}}{n}  \int_{0}^{\infty}{    e^{-\frac{n a_\mathrm{N} x^{-\alpha_\mathrm{N}} T \sigma^2}{C_\mathrm{N} G_\BS G_\MS}-\overline{\Upsilon \vphantom{(a)}}^3_n(T,x)-\overline{\Upsilon \vphantom{(a)}}^4_n(T,x)}     f_{r_0}^\mathrm{N}(x) dx}, \\
					& \hspace{20pt} \times \left[\sum_{n=1}^{N_\mathrm{N}}  \left(-1\right)^{n+1}  \dbinom{N_\mathrm{N}}{n} \int_{0}^{\infty}{    e^{-\overline{\Upsilon \vphantom{(a)}}^3_n(T,x)-\overline{\Upsilon \vphantom{(a)}}^4_n(T,x)}     f_{r_0}^\mathrm{N}(x) dx}\right]^{\Nms-1}, \nonumber
					\end{align}
					where $\overline{\Upsilon \vphantom{(a)}}^1_n(T,x)$ and $\overline{\Upsilon}^2_n(T,x)$ are similar to ${\Upsilon}^1_n(T,x)$ and ${\Upsilon}^2_n(T,x)$ in \eqref{eq:Up1_U}-\eqref{eq:Up2_U} but with the difference that  $\tau_\mathrm{L}(g,x)$ is replaced by $\frac{P_T G_\BS G_\MS C_L x^{-\alpha_L}}{T}$, and $a$ replaced by $a_L=N_L (N_L!)^{-\frac{1}{N_L}}$. A similar note applies for $\overline{\Upsilon \vphantom{(a)}}^3_n(T,x)$ and $\overline{\Upsilon \vphantom{(a)}}^4_n(T,x)$.
					\label{th:thm2}
		\end{theorem}
		\begin{proof}
			The proof is similar to that in Appendix \ref{app:Full_Pilot_UB}, and is omitted due to space limitations.
		\end{proof}}
	\item{\textbf{Effective Achievable Rate:}
		The full pilot reuse approach assumes that downlink control pilots are fully reused among all BSs, i.e., no specific sequence is needed to differentiate between the cells. Hence, the number of pilots $\Pi_\mathrm{p}\left(\lambda\right)=1$, and the data transmission efficiency is
		\begin{equation}
		\eta=\left(1-\frac{\Nbs \Nms}{L_\mathrm{C}}\right)^+.
		\end{equation}
		The effective achievable rate can be calculated using \eqref{eq:eff_Rate_F}
		\begin{equation}
		R_\mathrm{eff}= \left( 1-\frac{ N_\BS N_\MS}{L_\mathrm{C}} \right)^+ \left[\frac{1}{\ln(2)} \int_{y=T_\mathrm{th}}^{T_\mathrm{max}} \frac{{P}^{\text{Reuse}}_{\mathrm{c}}\left(y\right)}{y+1} dy + \log_2\left(1+T_\text{th}\right) P_\mathrm{c}^\mathrm{Reuse}\left(T_\text{th}\right) \right],
		\end{equation}
		which can be bounded using the results in Theorem \ref{th:thm1} and Theorem \ref{th:thm2}.
		}
	\item{\textbf{Verification of Analytical Results:}
		\begin{figure} [t]
			\centerline{
				\includegraphics[width=.6\columnwidth]{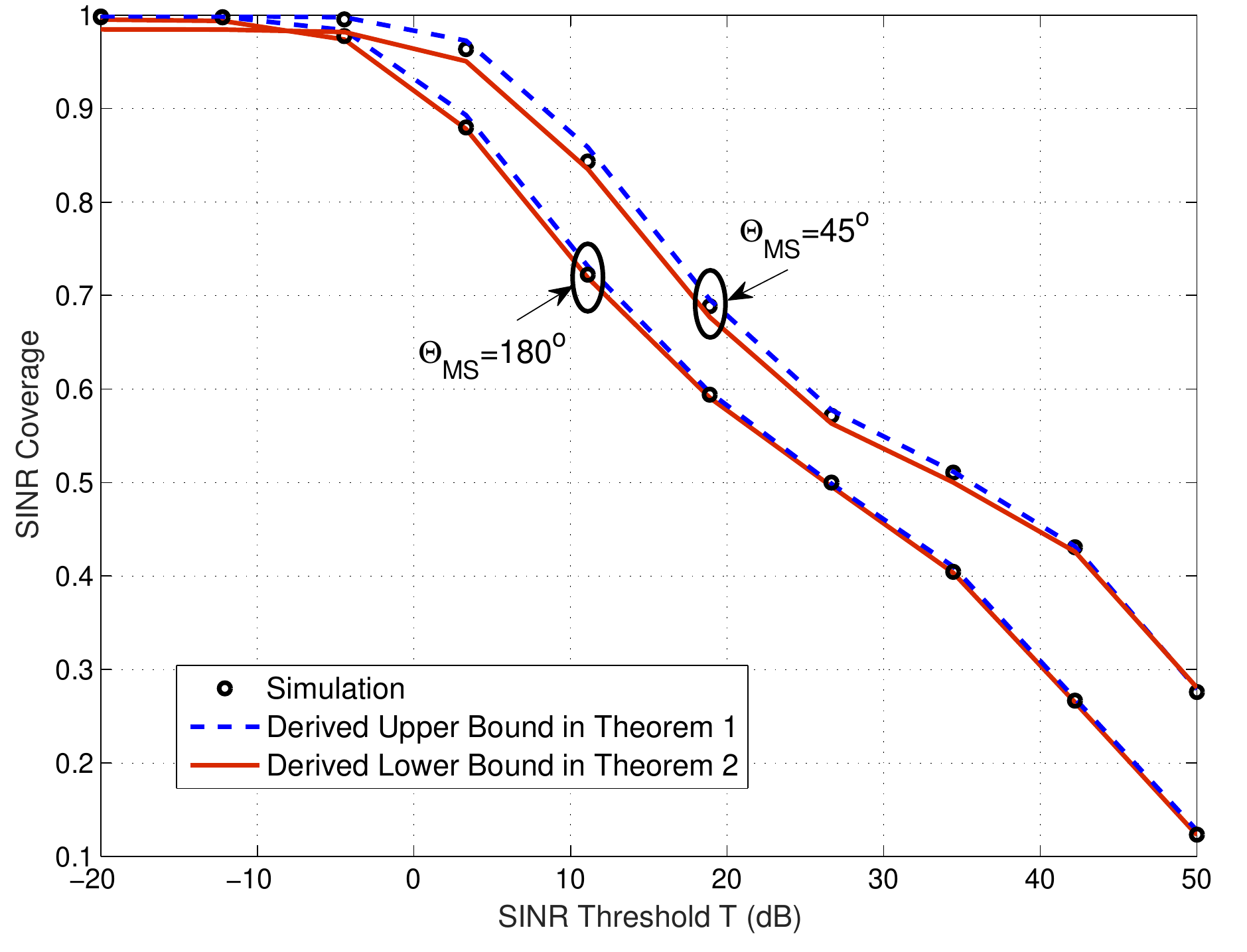}
			}
			\caption{The upper and lower bounds on the SINR coverage probability derived in Theorem \ref{th:thm1} and \ref{th:thm2} are plotted for different values of MS beamforming beamwidth $\Theta_\MS=180^\circ, \Theta_\MS=45^\circ$, and compared with the simulation results. The BSs use beamforming vectors of beamwidth $\Theta_\BS=6^\circ$, and are distributed as a PPP with a density $\lambda=\frac{1}{\pi R_c^2}$,  $R_c=50$m.  }
			\label{fig:Analysis}
		\end{figure}
		To validate the accuracy of the derived bounds in Theorems \ref{th:thm1}-\ref{th:thm2}, and examine their tightness, we compare the SINR coverage of the derived expressions with that given by numerical simulations in \figref{fig:Analysis}. In this figure, we adopt the mmWave cellular system model described in \sref{sec:NetModel} with the beam training/association model in \sref{sec:Assoc_Model}, assuming a network operating at a carrier frequency $f_c=28$ GHz, and with a bandwidth $100$ MHz. The BSs are distributed as a PPP with a density $\lambda=\frac{1}{\pi R_c^2}$, and an average cell radius of $R_c=50$ m. The LOS/NLOS path-loss exponents are $\alpha_\mathrm{L}=2.5, \alpha_\mathrm{N}=4.5$, and the LOS/NLOS Nakagami fading parameters are $N_\mathrm{L}=2, N_\mathrm{N}=3$. BSs are assumed to transmit with an average power of $P_T=43$ dBm, and to employ a codebook of sectored beamforming vectors as described in \sref{subsec:Ant_Model} with $N_\BS=64$, i.e., the sectored patterns have beamwidth $\Theta_\BS \approx 6^\circ$. Two codebooks with different beamwidths are assumed for the typical user, namely $N_\MS=8 \left(\Theta_\MS=45^\circ\right)$	and $N_\MS=2 \left(\Theta_\MS=180^\circ\right)$. Figure \figref{fig:Analysis} shows that the derived upper and lower bounds are tight for different beamwidth values of the MS combining patterns, which indicates that they both can be used as good approximations for the SINR coverage probability.
		}
\end{itemize}

\section{Discussion: Initial Beam Association Impact and Design Insights} \label{sec:Results}
In this section, we evaluate the performance of mmWave cellular systems with beam association overhead. We also provide some insights into how this beam association phase can be designed to achieve good system performance. First, we investigate the impact of the initial beam training/association phase on mmWave cellular systems in \sref{subsec:Impact}, and attempt to draw conclusions on how dense the BSs density and how narrow the BS/MS beams should be. Then, we consider the control pilot reuse factor in \sref{subsec:Pilots_Sim}, and explore how it can be optimized to achieve better effective rates for mmWave networks. Finally, we compare the performance of exhaustive search versus hierarchical search for the beam training phase in \sref{subsec:Hierarchical}, and show when hierarchical search can be beneficial and when it may not be needed.

The simulation results in this section adopt the system and beam association models described in \sref{sec:NetModel}-\sref{sec:Assoc_Model}. The system carrier frequency and bandwidth are $f_c=28$ GHz, $BW=100$ MHz, the LOS/NLOS propagation and fading parameters are $\alpha_\mathrm{L}=2.5, \alpha_\mathrm{N}=4.5, N_\mathrm{L}=2, N_\mathrm{N}=3$, the SINR  with maximum supportable MCS is $T_\mathrm{max}=\infty$, and the MCS SINR threshold for reliable transmission is $T_\mathrm{th}=0$ dB. The density of the co-pilot BSs is chosen according to the definition in \sref{subsec:SINR_Perf} such that $\epsilon_1=\epsilon_2=0.01$. The adopted antenna model has the sectored pattern explained in \sref{subsec:Ant_Model}. For the sake of accurate results and insights, we use the following expressions for the gain values $G_\BS, g_\BS$ in terms of the beamwidth $\Theta_\BS$ and the front-to-back power ratio $\gamma$:  
\begin{equation}
D^\ell_\BS\left(n\right)=\begin{cases} G_\BS= \frac{2 \pi}{\Theta_\BS} \frac{\gamma}{\gamma+1}, 
& \text{In-sector} \\ g_\BS=\frac{2 \pi}{2 \pi - \Theta_\BS} \frac{1}{\gamma+1},
& \text{Out-of-sector}
\end{cases}
\end{equation}
We further assume that the forward-to-backward power ratio $\gamma$ is given by $\gamma=\frac{2 \pi}{C_0 (2 \pi - \Theta_\BS)}$ for some constant $C_0$. This model is developed to achieve three important conditions: (1) the power conservation constraint which requires the total transmitted power from the antenna pattern to be constant and not a function of the beamwidth, by having $G_\BS \frac{\Theta_\BS}{2 \pi}+ g_\BS \frac{2 \pi - \Theta_\BS}{2 \pi}=1$, (2) the continuity condition for the gains in terms of the beamwidth, which implies that when $\Theta_\BS=1$, we should recover the omni-directional pattern with $G_\BS= 2 \pi$, and (3) the proper scaling condition which means that the main-lobe and side-lobe gains should have appropriate scaling behavior with the beamwidth. To do that, we modeled the front-to-back power ratio $\gamma$ to mimic the gain-beamwidth scaling behavior of the uniform linear arrays (ULA's) with beamsteering-based beamforming \cite{VanTrees2002,Balanis2005}.

\subsection{Impact of Beam Association on Millimeter Wave Cellular Performance} \label{subsec:Impact}
 Prior work on mmWave cellular system performance focused on studying the coverage and rate of mmWave networks assuming perfect beam alignment and putting no constraints on the minimum SINR for reliable MCS \cite{Bai2015,Singh2015}. The work in \cite{Bai2015,Singh2015,Bai2014}, and \cite{Pi2011} among others, suggested that two important features of mmWave cellular networks will be the deployment of dense BSs and the employment of narrow beams at both the BSs and mobile users. With beam association and under reliability SINR constraints, however, these features may be questionable. It is, therefore, important to determine how the beam training/association phase and the reliability SINR constraints impact the key features of mmWave cellular systems, namely, the deployment of dense BSs and the employment of narrow beams.

\subsubsection{How Dense Should MmWave Networks Be?} 
\begin{figure}
	\centerline{
		\subfigure[][{LOS range $\mu=50$m}]{
			\includegraphics[width=.55\columnwidth]{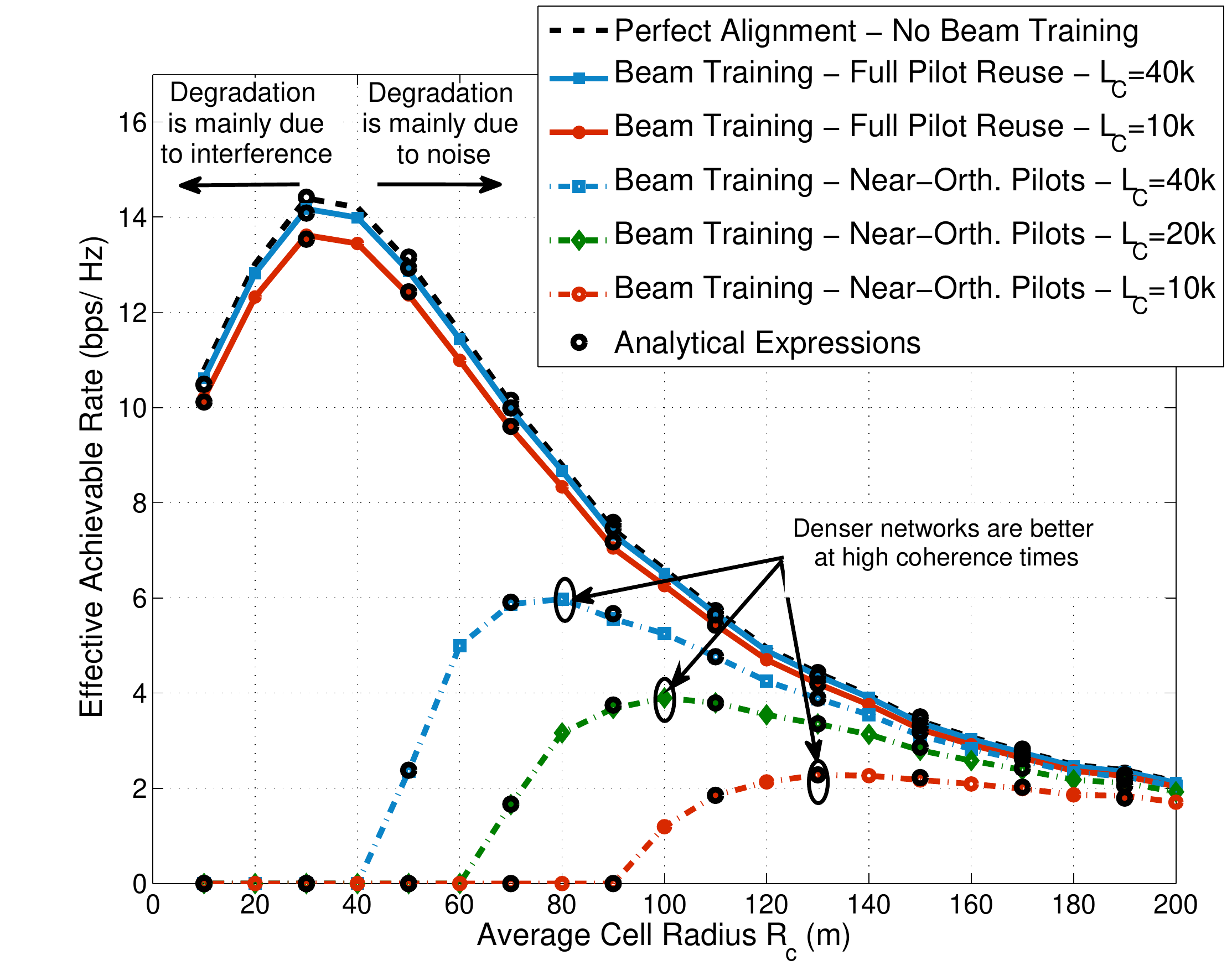}
			\label{fig:Den_50}	
		}
		\subfigure[][{LOS range $\mu=150$m}]{
			\includegraphics[width=.55\columnwidth]{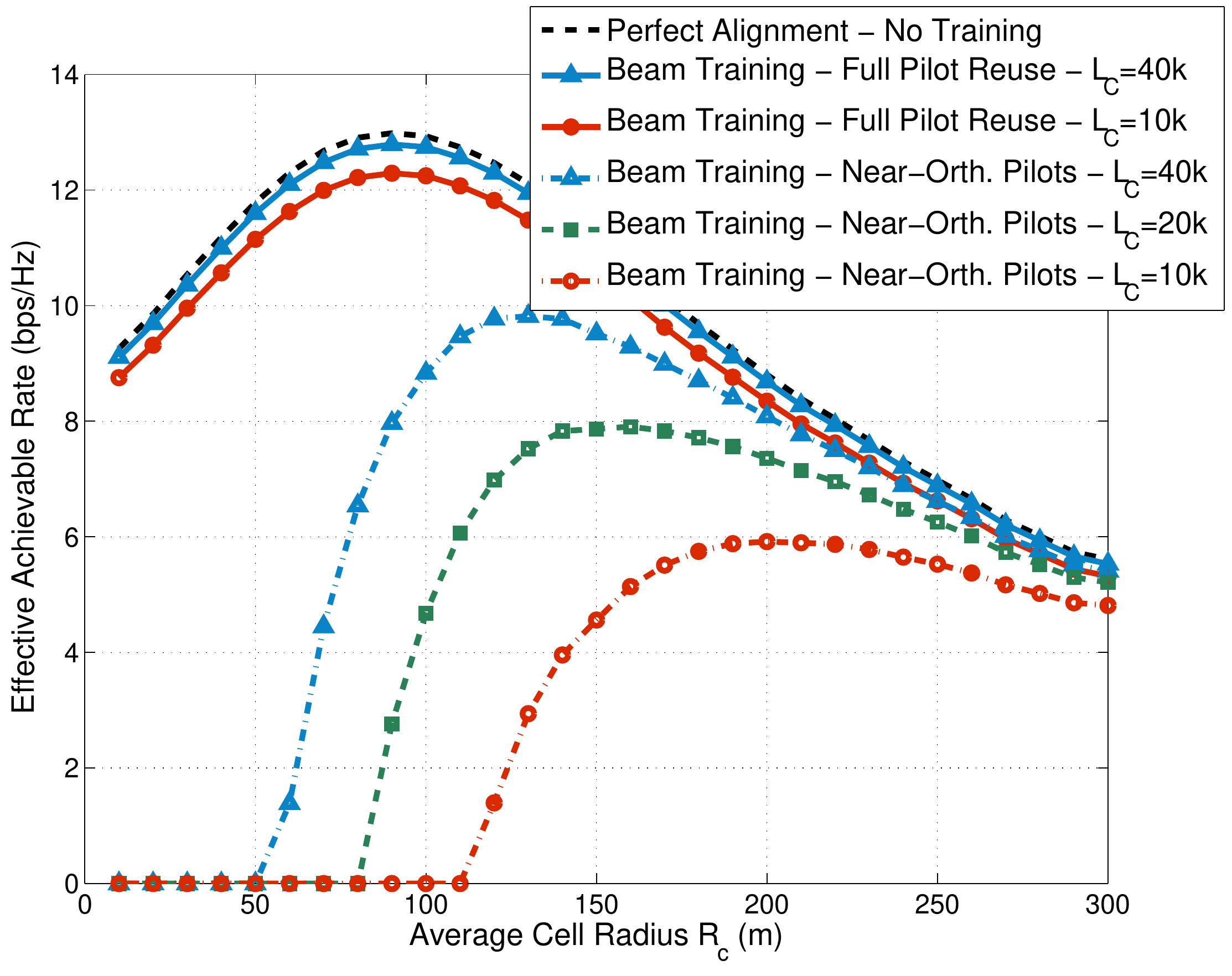}
			\label{fig:Den_150}
		}}
		\caption{The effective achievable rates of mmWave cellular networks under different beam association models are compared for different values of coherence block length with LOS range $\mu=50$m in (a), and $\mu=150$m in (b). The BSs and MSs use beamforming/combining vectors of beamwidth, $\Theta_\BS=6^\circ$ and $\Theta_\MS=45^\circ$.}
		\label{fig:Den}
	\end{figure}

\begin{figure} [t]
	\centerline{
		\includegraphics[width=.6\columnwidth]{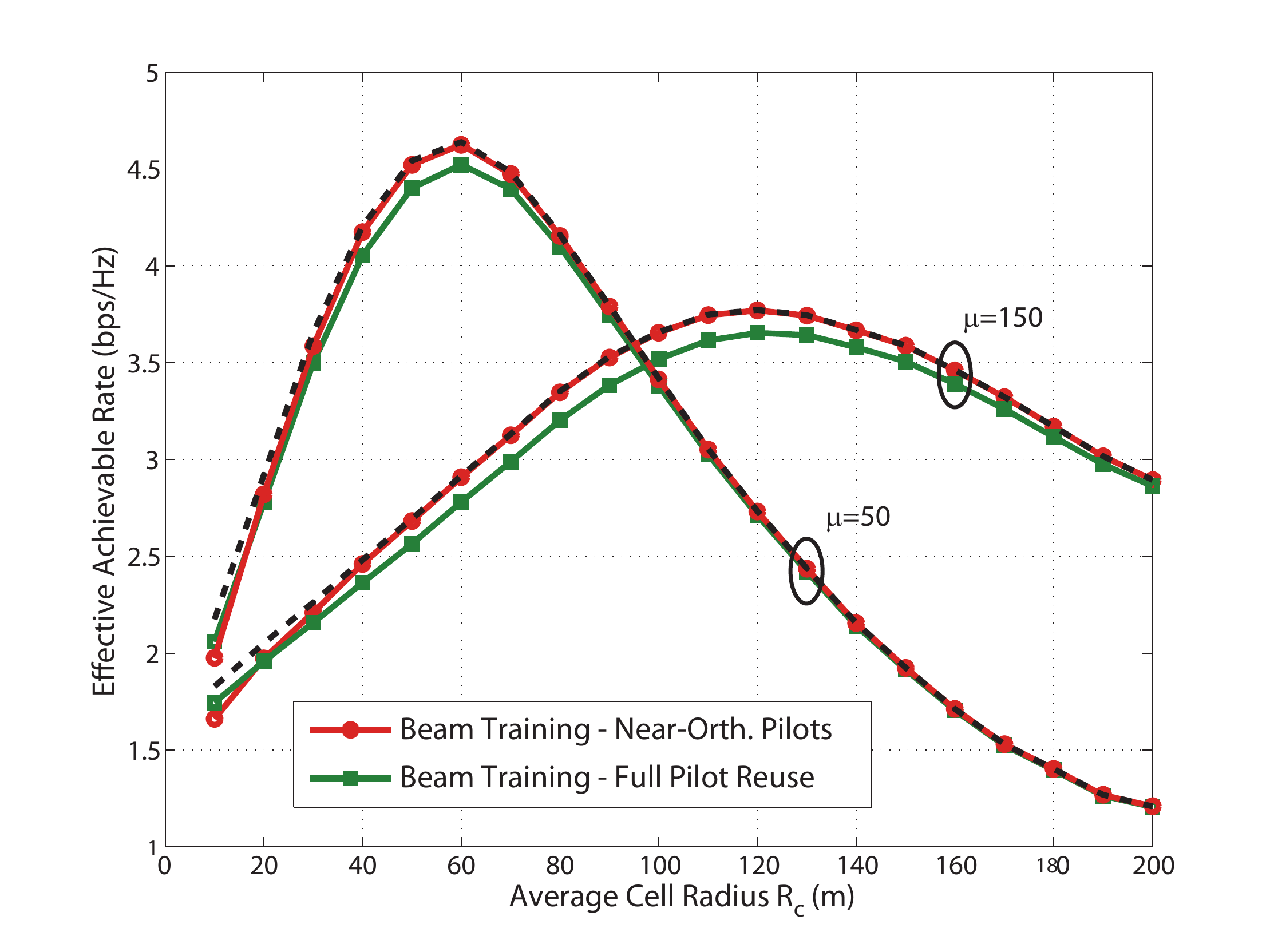}
	}
	\caption{The effective achievable rates of mmWave cellular networks under different beam association models are compared for different values of LOS ranges, with a coherence block length  of $L_C=70$k symbols. The BSs and MSs use beamforming/combining vectors of beamwidth $\Theta_\BS=45^\circ, \Theta_\MS=180^\circ$, respectively. }
	\label{fig:Den_Wide}
\end{figure}

We plot the effective reliable rate of mmWave cellular networks, defined in \eqref{eq:eff_Rate_F}, versus the average cell radius $R_c=\frac{1}{\sqrt{\pi\lambda}}$ for different scenarios in \figref{fig:Den} and \figref{fig:Den_Wide}. In \figref{fig:Den_50}, we consider mmWave networks with BSs and MSs employing narrow beams $\Theta_\BS = 6^\circ, \Theta_\MS = 45^\circ$, and with a small LOS range $\mu=50$m. First, we note that the effective reliable rate for this system with perfect beam alignment has an optimum average cell radius of about $30$m (or a BS density of $\lambda \approx 350$ BS/km$^2$). With larger cell radius, the effective rate degrades mainly because of the noise power compared to the signal power. The performance also degrades when the cell radius becomes very small mainly because of the LOS interference. As for the beam training/association impact, we compare the performance with near-orthogonal downlink control pilots and full pilot reuse models for different coherence block lengths. Results show that the performance with near-orthogonal control pilots suggests larger average cell radius (more sparse networks), especially at low coherence block lengths. This is mainly because it requires higher training overhead at dense networks as indicated in \eqref{eq:Rate_SINR}. Results also illustrate that much better performance can be achieved with full pilot reuse, which assumes $\delta_c=1$. This implies that the impact of out-of-cell interference is relatively small, and it yields a small SINR penalty, in this system model with narrow BS and MS beams. Similar performance behavior is experienced in \figref{fig:Den_150} when a larger LOS range $\mu=150$m is assumed. Results, though, show that the optimal average cell radius with perfect beam alignment is shifted towards more sparse networks. This is a result of the larger LOS range $\mu$, which causes more interference links to be LOS with dense networks, which in turns shifted the interference and noise limited regimes towards more sparse networks.

Finally, we plot the performance with wider BS and MS beams, $\Theta_\BS=45^\circ, \Theta_\MS=180^\circ$, in \figref{fig:Den_Wide} for both LOS ranges $\mu=50, \mu=150$, and considering a coherence block length $L_\mathrm{C}=70$k symbols. The beam association model with near-orthogonal pilots yields a better effective rates in this case, compared to that with full pilot reuse, and achieves a very close performance to the perfect beam alignment case. The performance difference between the two beam association models, though, is relatively small. These results show that wider beams have higher impact on the beam association phase in terms of the SINR penalty. These wide beams, however, require less training overhead, which makes the performance of beam association with near-orthogonal pilots very close to the perfect beam alignment case.

The results in \figref{fig:Den} and \figref{fig:Den_Wide} suggest that the average cell radius of mmWave cellular networks should be around the LOS range $\mu$. Further, these results indicate that the average cell radius, which maximizes the effective reliable rate, with beam association overhead can be made close to that with perfect beam alignment if proper beam association strategy is adopted. This selection  depends on many parameters including the beamwidth of the employed beamforming/combining beams at the BSs and mobile users.
\subsubsection{How Narrow Should the Beams Be?}
\begin{figure}
	\centerline{
		\subfigure[][{MS beamwidth $\Theta_\MS=45^\circ$}]{
			\includegraphics[width=.515\columnwidth]{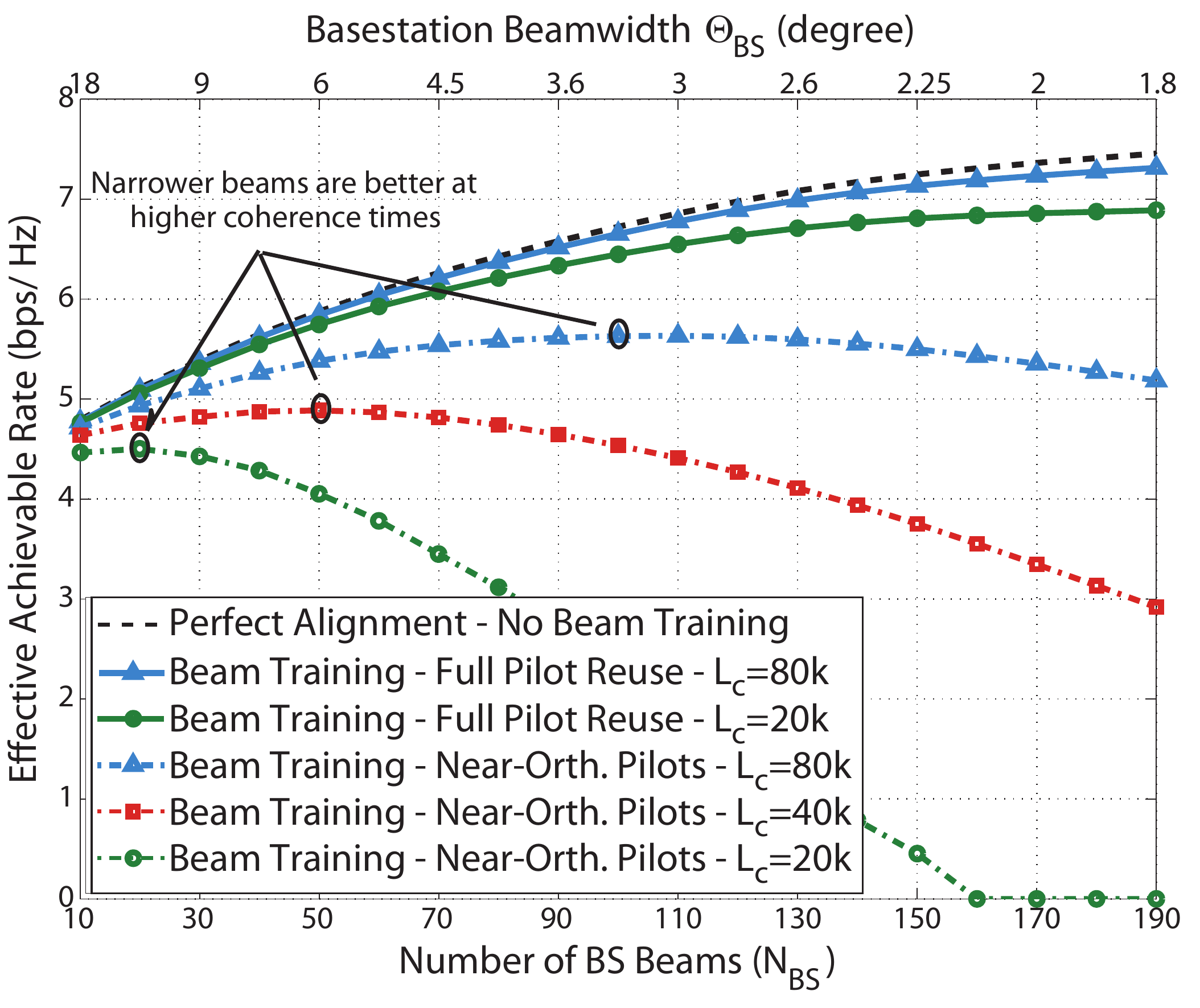}
			\label{fig:BW_8}	
		}
		\subfigure[][{MS beamwidth $\Theta_\MS=180^\circ$}]{
			\includegraphics[width=.49\columnwidth]{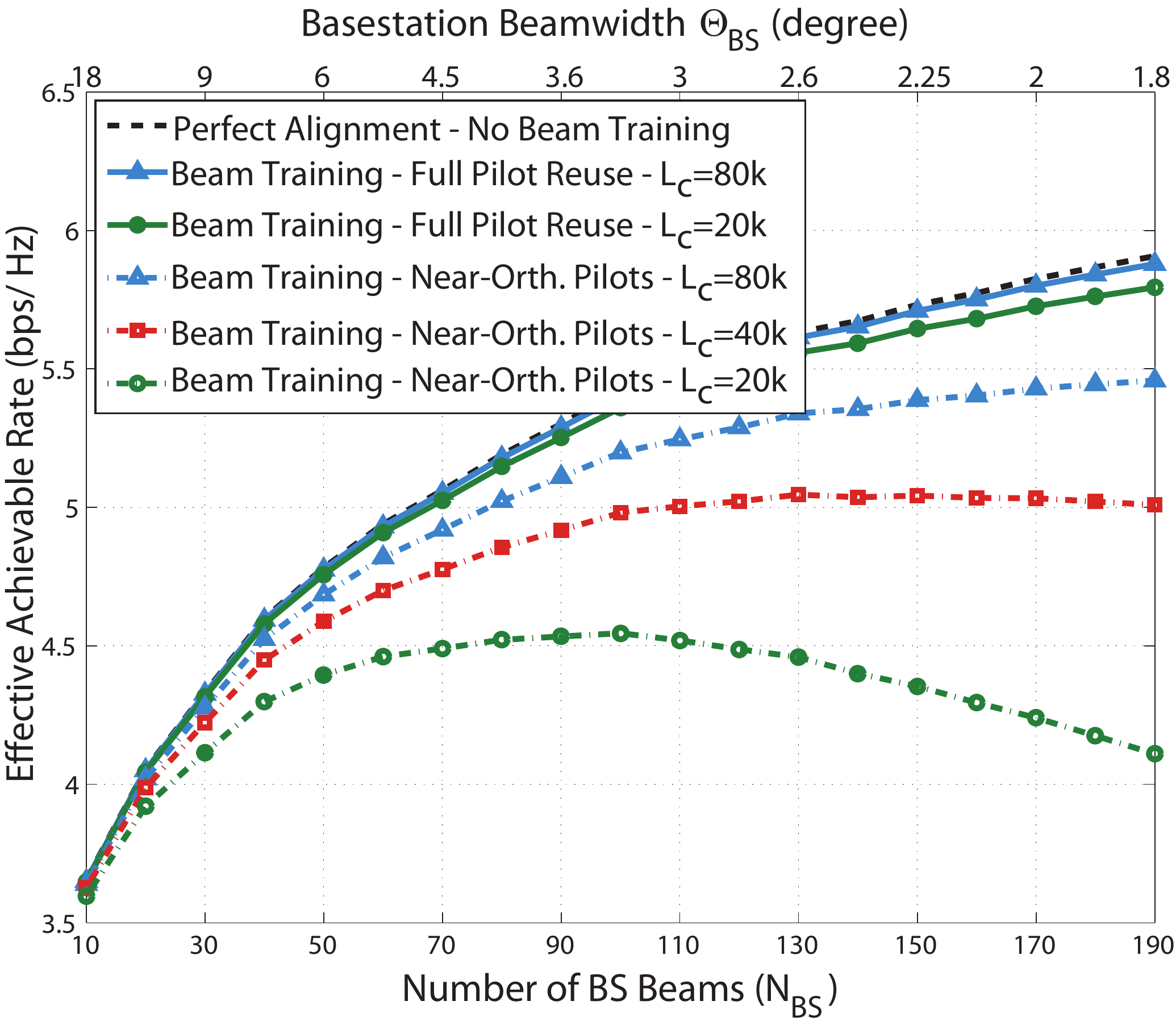}
			\label{fig:BW_2}
		}}
		\caption{The effective achievable rates of mmWave cellular networks under different beam association models are compared for different values of coherence block length with LOS range $\mu=50$m.} 
		\label{fig:BW}
	\end{figure}
To get some insights into the answer of this question, we plot the effective reliable rate versus the BS beamforming beamwidth in \figref{fig:BW} assuming a network with an average cell radius $R_\mathrm{c}=100$m and a LOS range $\mu=50$m. In \figref{fig:BW_8}, the mobile users are assumed to employ combining vectors with a small beamwidth $\Theta_\MS=45^\circ$. The results show that when perfect beam alignment is assumed, very narrow BS beams $\Theta_\BS < 2^\circ$ can give high effective reliable rates. When beam association overhead is considered with near-orthogonal pilots, a wider beamwidth is better as very narrow beams require much more training overhead as given by \eqref{eq:Rate_SINR}. The effective rate with full pilot reused based beam association is better than that with near-orthogonal pilots, and approach the perfect beam alignment performance. Similar behavior is experienced in \figref{fig:BW_2} when wider beams are assumed to be deployed by the MSs $\Theta_\MS=180^\circ$. In \figref{fig:BW_2}, narrower beams may be better for  the near-orthogonal pilot model compared with \figref{fig:BW_8}, as the training overhead is reduced with wider MS beams.  

From the results in \figref{fig:BW}, we conclude that if the control pilot reuse factor is well designed, then the effective reliable rate improves with narrow beams. This is because of the possible high beamforming gains when narrow beams are employed.  It is worth noting here that our approach to determine the good beamwidth for the beamforming vectors focused on the effective reliable rate metric with the system and beam association models described in \sref{sec:NetModel}, \sref{sec:Assoc_Model}. This approach, however, did not account for other practical limitations on the antenna design or other impairments on the beamforming design that may need to be considered for a more accurate decision on the optimal beamwidth. This remains an interesting topic for future work. 
\subsection{Orthogonal or Reused Pilots?} \label{subsec:Pilots_Sim}
\begin{figure}
	\centerline{
		\subfigure[][{Different numbers of antennas}]{
			\includegraphics[height=190pt,width=.5\columnwidth]{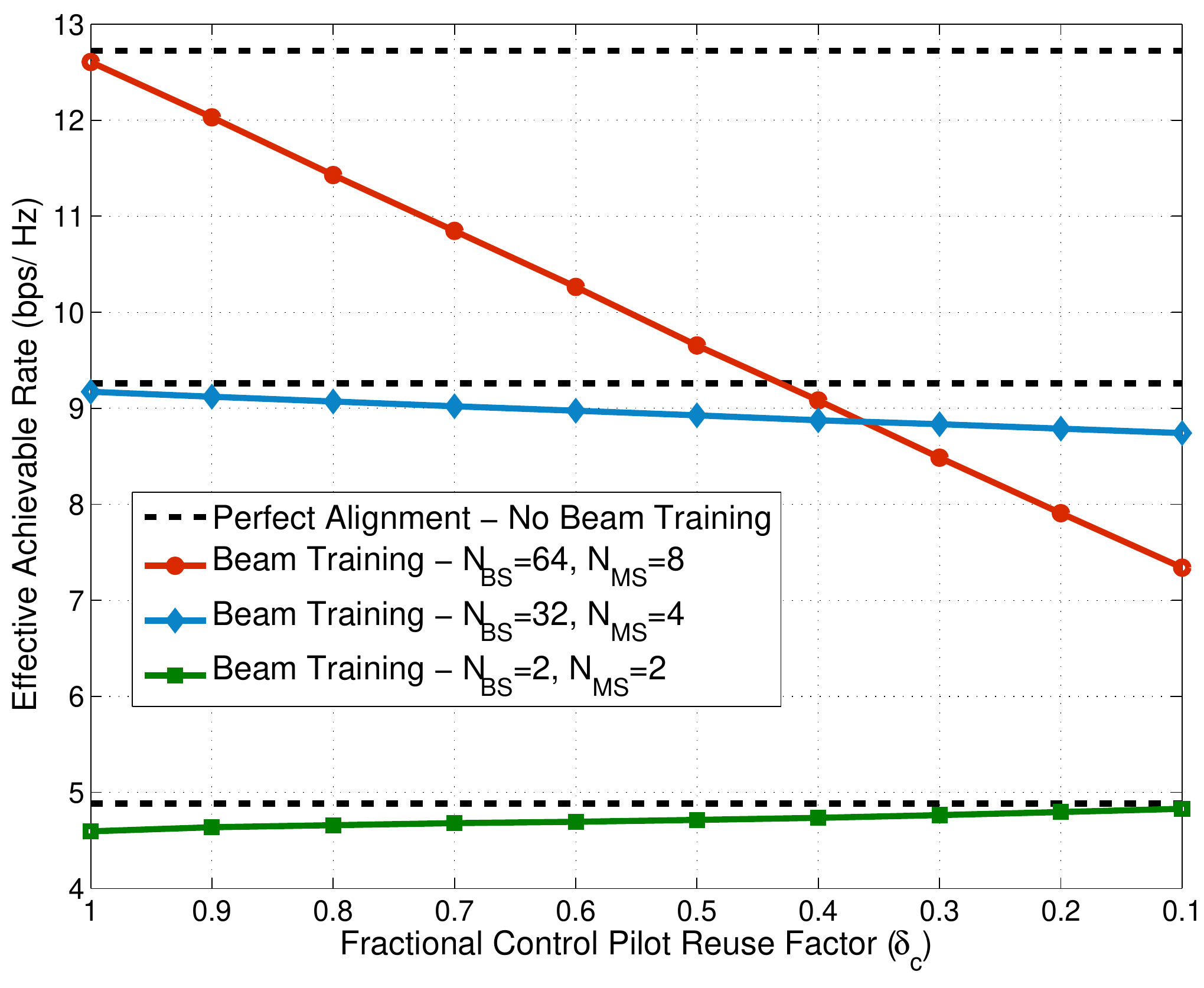}
			\label{fig:Pilots_reuse}	
		}
		\subfigure[][{Different LOS ranges}]{
			\includegraphics[height=190pt,width=.5\columnwidth]{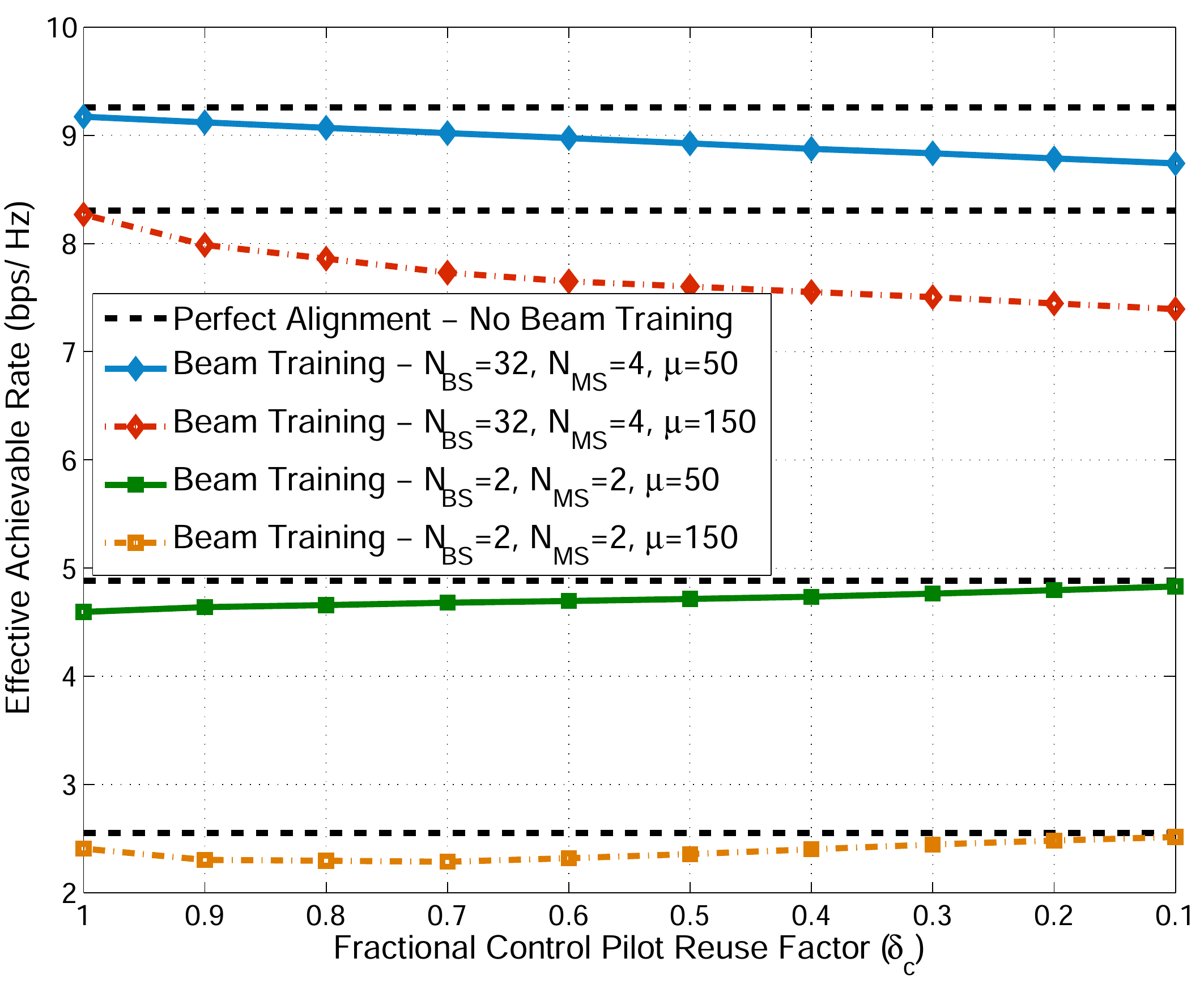}
			\label{fig:Pilots_Beta}
		}}
		\caption{The impact of pilot reuse factor on the performance of mmWave cellular networks is shown for different values of beamforming/combining codebook sizes (or equivalently beamwidths), and with an average cell radius $R_c=50$m.}
		\label{fig:Pilots}
\end{figure}
From the results in \figref{fig:Den}, \figref{fig:Den_Wide}, we noticed that each of the two adopted beam association models, namely, with near-orthogonal pilots and with full pilot reuse, is better in certain scenarios. Those two models, however, represent two extreme cases for the control pilot reuse factor, i.e., with $\delta_c=\delta_{c,\text{min}}$ and with $\delta_c=1$. To examine the impact of different pilot reuse factors, we plot the effective reliable rate in \figref{fig:Pilots} versus the control pilot reuse factor, assuming a network with an average cell radius $R_\mathrm{c}=50$m, and a coherence block length $L_\mathrm{c}=70$k symbols. In \figref{fig:Pilots_reuse}, the performance of three different setups of the BS/MS codebook sizes (or equivalently beamwidths) are compared. These results show that adopting small pilot reuse factors yields better performance when the BSs and MSs use wide-beamwidth beamforming beams, $\Theta_\BS=\Theta_\MS=180$ at $N_\BS=N_\MS= 2$.  With the other two setups which use narrower beams, full pilot reuse achieves better performance. The main reason behind this behavior is the impact of the out-of-cell interference on the beam alignment error, which increases with wide beams. Alleviating it with a small pilot reuse factor is better in this case. From \figref{fig:Pilots_reuse}, the points with maximum effective rates are mainly located in the extreme cases of full pilot reuse or very small pilot reuse factors. Similar behavior with different rate values is experienced for different values of the LOS range, $\mu$, as shown in \figref{fig:Pilots_Beta}. The results in \figref{fig:Pilots} recommend full pilot reuse factors to be adopted by mmWave networks, unless the employed data transmission beams have very wide beamwidths.
 
\subsection{Exhaustive or Hierarchical Search?} \label{subsec:Hierarchical}
\begin{figure} [t]
	\centerline{
		\includegraphics[width=.6\columnwidth]{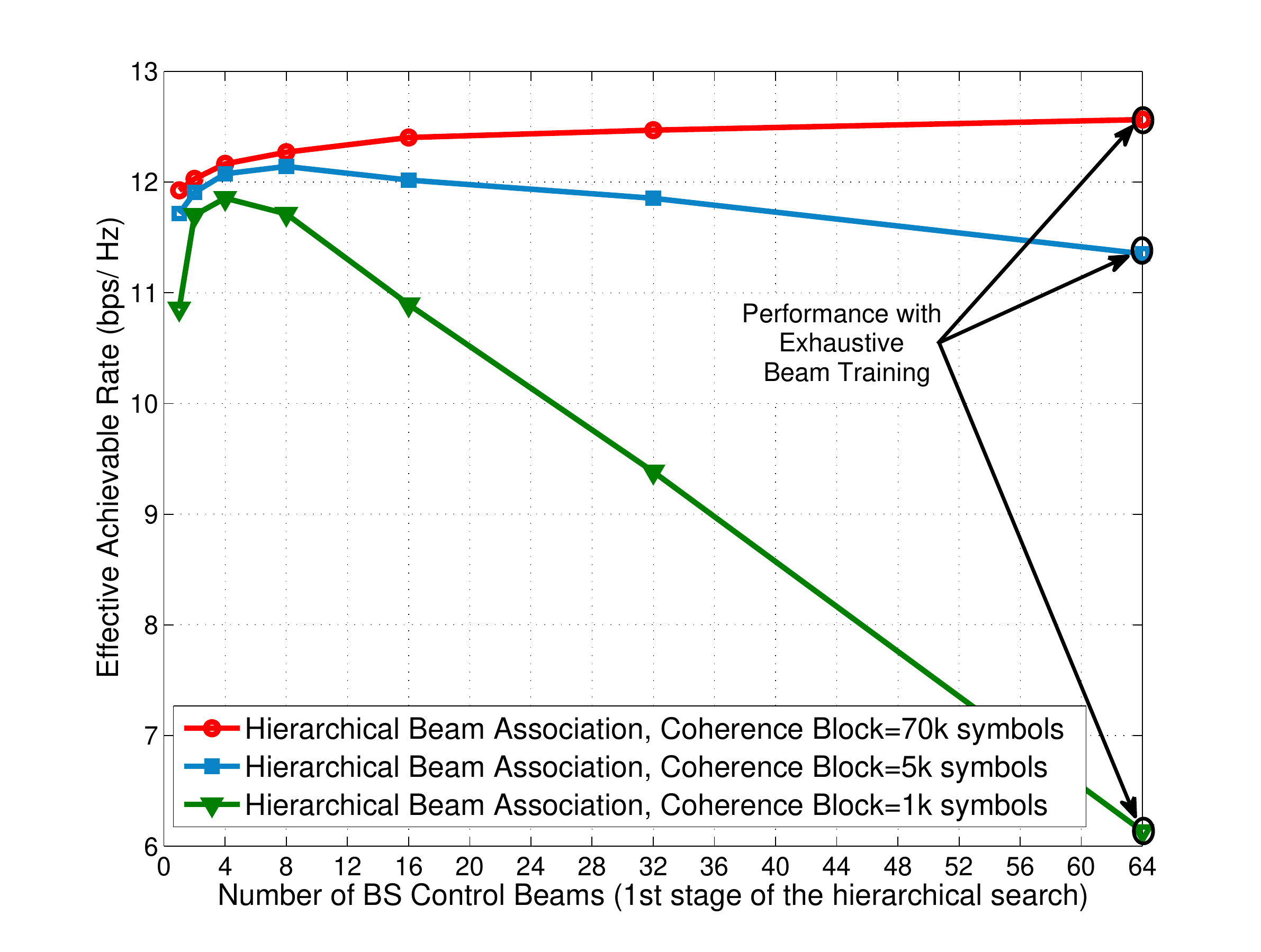}
	}
	\caption{The effective reliable rates of mmWave networks with hierarchical beam search are plotted for different coherence block lengths. The BS is assumed to first find the best control beamforming vector in a codebook with the given number of beams on the x-axis, before refining this beam to have a data transmission beam with a beamwidth $\Theta_\BS = 6^\circ$. The MS is assumed to have use combining beams with a beamwidth $\Theta_\MS=45^\circ$. The solution with $64$ control beams represent the exhaustive search solution, i.e., the direct search with the narrowest beams, as no refining is needed in this case. }
	\label{fig:Hier}
\end{figure}

In this paper, we focused on an exhaustive search for beam sweeping. This search directly results in the \textit{narrow} data transmission beam pair. An alternative to tune this beam pair is through hierarchical search \cite{Desai2014,Wang2009}, where the data beamforming vectors are adjusted via two stages: (i) a control beam search stage in which beamforming vectors with wide beamwidths are used to find the best control beamforming/combining pairs, and (ii) a data beam search stage in which the best control beamforming/combining vectors resulted from the control beam search stage are further refined to choose the best narrow beamforming/combining vectors for the data transmission \cite{Desai2014}. One main advantage of the hierarchical search is the smaller training overhead compared with exhaustive search. Considering the insights from \figref{fig:Pilots}, though, wide beams experience more training misalignment errors compared with narrow beams. This is because of the small desired beamforming gain and the large impact of out-of-cell interference when wide beams are employed. To alleviate this error, small pilot reuse factors, i.e., more pilots, may be needed in the control/wide beam search phase. This changes the overall training overhead of hierarchical search and makes it interesting to investigate its performance compared with exhaustive search. Note that exhaustive search with narrow beams normally experience small misalignment errors with full pilot reuse as shown in \figref{fig:BW}-\ref{fig:Pilots}. Next, we evaluate the performance of mmWave networks using numerical simulations, comparing the exhaustive and hierarchical search strategies.

For the hierarchical search, if $\overline{N}_\BS, \overline{N}_\MS$ represent the codebook sizes of the BS and MS beamforming vectors in the first (wide beams) search phase, with , $\overline{\delta}_c$ as the pilot reuse factor. Then, the effective reliable rate becomes
%
\begin{equation}
R_\mathrm{eff}=\left(1-\frac{\overline{N}_\BS \overline{N}_\MS \frac{1} {\overline{\delta}_c} + \frac{N_\BS}{\overline{N}_\BS} \frac{N_\MS}{\overline{N}_\MS}\frac{1}{\delta_c}}{L_\mathrm{c}}\right)^+ \bbE\left[\log_2\left(1+\max\left(\mathsf{SINR}, T_\mathrm{max}\right)\right) \mathbbm{1}_{\mathsf{SINR} \geq T_\mathrm{th}} \right].
\end{equation}

In \figref{fig:Hier}, the reliable effective rate of a mmWave network with average cell radius of $R_c=50$m and a LOS range $\mu=50$m, is plotted versus the control beams codebook size $\overline{N}_\BS$. The final BS data beams are assumed to have $N_\BS=64$, and the MS is assumed to have $N_\MS=\overline{N}_\MS=8$. Further, the pilot reuse factors are numerically optimized with $\delta_c, \overline{\delta}_c \in \left[0.1, 1\right]$. Note that the solution $\overline{N}_\BS$ represents the exhaustive search solution, as the control beams have the same beamwidth of the data beams in this case. The results in \figref{fig:Hier} show that hierarchical search achieves better effective rates when the coherence block length is relatively small. If the coherence block length is large enough, though, exhaustive search provides a similar or better performance. 
\subsection{Evaluation with Uniform Arrays and Multi-Path Channels} \label{subsec:MP}
\begin{figure} [t]
	\centerline{
		\includegraphics[width=.6\columnwidth]{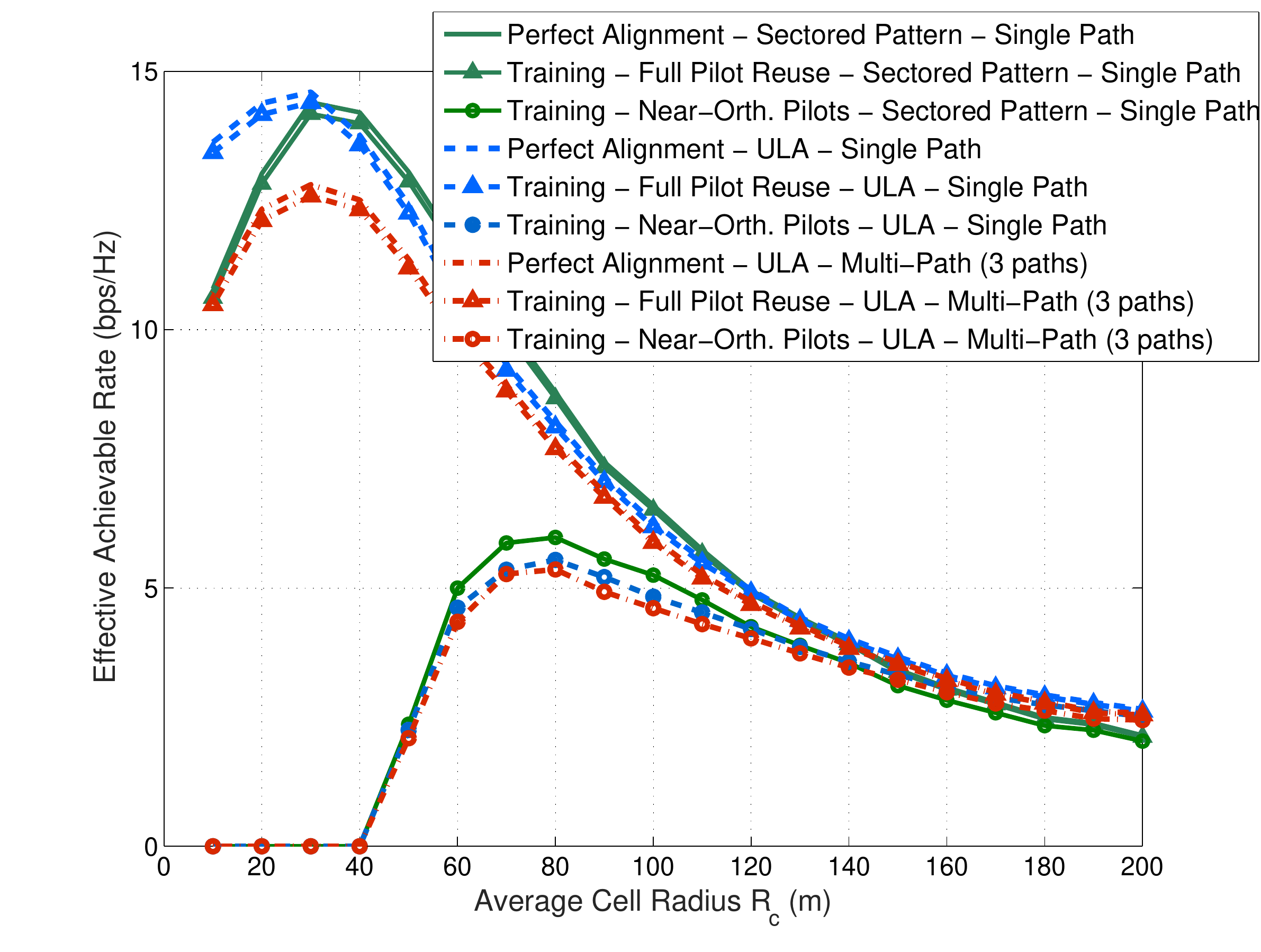}
	}
	\caption{The effective achievable rates of mmWave cellular networks for three setups of the channels and the antenna models. The figure shows that the general performance trends are similar with some scaling differences between the three models. }
	\label{fig:ULA}
\end{figure}


To gain analytical tractability, we made two main assumptions in the adopted system model in \sref{sec:NetModel}, namely, the single-path channel and zero MS side-lobe assumptions. In this subsection, we test the fragility of these assumptions by simulations. In \figref{fig:ULA}, the network model of \figref{fig:Den_50} is adopted for three different channel/antenna scenarios: (i) when the channels have single-paths and the BSs/MSs have the sectored beam patterns in \sref{sec:NetModel}, (ii) when the channels have single-paths and the BSs-MSs employ ULA's with half-wavelength, 64 antennas for the BS, 8 antennas for the MS, and (iii) when the channels between the MS and each BS consists of 3 clusters, each contributes with a single ray (one of the paths can be LOS), and the BSs/MSs antennas employ the same ULA's. \figref{fig:ULA} shows that the performance of mmWave systems with perfect beam alignment and with beam association has similar trends for the three scenarios, with some scaling differences. It also illustrates that the beam association models behave similarly in the three scenarios. This can be verified for the performance trends versus the beamforming beamwidth, and the pilot reuse factor. We  conclude that the assumptions of single-path links and sectored antennas, made for analytical tractability, are reasonable.

\section{Conclusions}\label{sec:Conclusion}
In this paper, we evaluated the performance of mmWave cellular networks with beam training/association overhead. We defined and adopted a new performance metric called effective reliable rate that captures the impact of the beam association overhead and the minimum threshold on the SINR for reliable transmission. Using stochastic geometry, we analyzed the effective rate of mmWave networks for two special yet important cases, namely, when the control pilots are nearly orthogonal and when they are fully reused. The results allow us to make the following performance insights.
\begin{itemize}
\item{\textit{How dense should mmWave networks with narrow beams be?} The effective rate of a mmWave cellular network is maximized when the average cell radius is in the order of the LOS range. Further, for the same LOS range, the optimal cell radius increases as the beamwidth of the adopted beams increases. Finally, the effective reliable rate improves as the beamwidth decreases. }

\item{\textit{What is the impact of initial beam association on network performance?} Beam training/association has resource overhead and as SINR penalty on mmWave network performance. With narrow beams, the beam training SINR penalty is low while the resource overhead is relatively high. This is reversed when wide beams are employed. Results, though, showed that the performance of mmWave networks with beam association can approach that with perfect beam alignment if the downlink control pilot reuse is wisely adjusted based on the key system parameters.}

\item{\textit{Should downlink control pilots be reused?} Interestingly, full pilot reuse can normally yield better effective reliable rates in mmWave cellular systems, unless the beamwidth of the employed beamforming/combining vectors is very wide.}

\item{\textit{Should hierarchical or exhaustive search be performed?}  The system performance with hierarchical search exceeds that with exhaustive search only when the coherence block length is relatively small. If the coherence block length is large, exhaustive search gives a better performance thanks to its less beam misalignment error probability.}
\end{itemize}
For future work, it would be of interest to include more sources of misalignment errors, such as winds and antenna impairments, into the performance analysis. It would also be interesting to consider the impact of joint cell and beam association on the system performance.


\appendices
\section{}\label{app:Full_Pilot_UB}
\begin{proof}[Proof of Theorem \ref{th:thm1}]
	The CCDF of the data transmission SINR, taking into consideration the impact of initial beam association, can be written as
\begin{small}
\begin{equation}
P_\mathrm{c}^\text{TP}(T) =\mathbb{P}\left(\left.\mathsf{SINR} > T \right| \mathsf{OBP}\right) \mathbb{P}\left(\mathsf{OBP}\right) + \mathbb{P}\left(\left.\mathsf{SINR} > T \right| \mathsf{SBP}\right) \mathbb{P}\left(\mathsf{SBP}\right).
\end{equation}
\end{small}
Now, we derive an upper bound on the SINR coverage probability conditioned on the $\mathsf{OBP}$ event, $P^\text{TP}_{\text{c}|\mathsf{OBP}}\left(T\right)=\mathbb{P}\left(\left.\mathsf{SINR} > T \right| \mathsf{OBP}\right)$. First, conditioning on the LOS and NLOS probabilities of the desired link, the conditional SINR coverage $P^\text{TP}_\text{c$|$$\mathsf{OBP}$}\left(T\right)$ can be written as
\begin{small}
\begin{equation}
P^\text{TP}_\text{c$|$$\mathsf{OBP}$}\left(T\right)=P^\text{TP}_\text{c$|$$\mathsf{OBP}$,$\mathrm{L}$}\left(T\right) A_\mathrm{L} + P^\text{TP}_\text{c$|$$\mathsf{OBP}$,$\mathrm{N}$}\left(T\right) A_\mathrm{N},
\end{equation}
\end{small}
where $A_\mathrm{L}$ and $A_\mathrm{N}$ are the probability that the MS is associated with LOS and NLOS BS, respectively, and $P^\text{TP}_\text{c$|$$\mathsf{OBP}$,$\mathrm{L}$}\left(T\right)$ and $P^\text{TP}_\text{c$|$ $\mathsf{OBP}$,$\mathrm{N}$}\left(T\right)$ are the corresponding conditional SINR coverage for the LOS and NLOS cases. The expressions of $A_\mathrm{L}, A_\mathrm{N}$ are the same as those derived in \cite[Section III-B]{Bai2015}. Next, given the adopted total power beam pair decision metric, the LOS conditional SINR coverage  $P^\text{TP}_\text{c$|$$\mathsf{OBP}$,$\mathrm{L}$}\left(T\right)$ can be written as
\begin{small}
\begin{equation}
P^\text{TP}_\text{c$|$$\mathsf{OBP}$,$\mathrm{L}$}\left(T\right)=\mathbb{P}\left(\left.\frac{P_T G_\BS G_\MS h_0^\mathrm{L} C_\mathrm{L} r_0^{-\alpha_\mathrm{L}} }{I_{\Phi_{1}}+\sigma^2} > T \right|P_T G_\BS G_\MS h_0^\mathrm{L}  C_\mathrm{L} r_0^{-\alpha_\mathrm{L}} + I_{\Phi_{1}} > I_{\Phi_{2}}, ..., I_{\Phi_{N_\MS}} \right).
\end{equation}
\end{small}
Further, using  the Slivnyak's Theorem \cite{Baccelli2009}, and given that the user is associated with a LOS BS, the conditional LOS SINR coverage can be written as
\begin{small}
\begin{equation}
 P^\text{TP}_\text{c$|$ $\mathsf{OBP}$,$\mathrm{L}$}\left(T\right) =
\int_{0}^{\infty}{\int_{0}^{\infty}{ \mathbb{P}\left( I_{\Phi_{1}} < \tau_L\left(g,x\right)\left| C_0\left(g,x\right) + \sigma^2 + I_{\Phi_{1}} > I_{\Phi_{2}}, ..., I_{\Phi_{N_\MS}} \right. \right) f_{h_0^L} \left(g\right) } f_{r_0}^L(x) dg dx}, \label{eq:Cond1}
\end{equation}
\end{small}
where $\tau_L(g,x)=P_T G_\BS G_\MS g x^{-\alpha_L} \frac{C_L}{T}-\sigma^2$, and $f_{h_0^L} \left(g\right), f_{r_0}^L(x)$ are the probability distribution functions of the channel gain $h_0^L$ and the distance to the associated BS $r_0$, respectively. We recall here that $h_0^L$ is modeled as a normalized gamma random variable with a parameter $N_L$, and its PDF function is therefore given by $f_\gamma(g,N,1/N)=\frac{N^N x^{N-1} e^{-N x}}{(N-1)!}$. Now, we note that conditional coverage probability in \eqref{eq:Cond1} can be upper bounded by
\begin{small}
\begin{equation}
P^\text{TP}_\text{c$|$ $\mathsf{OBP}$,$\mathrm{L}$}\left(T\right) \leq \int_{0}^{\infty}{\int_{0}^{\infty}{ \mathbb{P}\left(I_{\Phi_{1}} < \tau_L\left(g,x\right) \right) f_{h_0^L} \left(g\right) } f_{r_0}^L(x) dg dx}. \label{eq:Cond2}
\end{equation}
\end{small}%
Following \cite{Thornburg2015}, the CDF of the interference random variable $I_{\Phi_{1}}$ can be expressed as
\begin{small}
\begin{align}
 \mathbb{P}\left(I_{\Phi_{1}} <\tau_L\left(g,x\right) \right) &= \lim_{N \to \infty}\sum_{n=1}^{N} \dbinom{N}{n} \left(-1\right)^{n+1} \bbE_{\Phi_1}\left[e^{-\frac{a n}{\tau_L(g,x)}I_{\Phi_{1}} }\right], \\
 & \stackrel{(a)}{=} \lim_{N\to \infty} \sum_{n=1}^{N} \dbinom{N}{n} \left(-1\right)^{n+1} \bbE_{\Phi_{1}^L}\left[e^{-\frac{a n}{\tau_L(g,x)}I_{\Phi^L_{1}} }\right]
 \bbE_{\Phi^N_{1}}\left[e^{-\frac{a n}{\tau_L(g,x)}I_{\Phi^N_{1}} }\right],
\label{eq:Int1}
\end{align}
\end{small}
where $a=N (N!)^{-\frac{1}{N}}$, $N$ is an arbitrary large integer number, and (a) follows by assuming that the interfering BSs belong to independent LOS and NLOS PPP's. The Laplace transforms of the interference terms in \eqref{eq:Int1} can further be  written in closed forms as \cite{Bai2015}
\begin{small}
\begin{align}
& \bbE_{\Phi^L_{1}}\left[e^{-\frac{a n}{\tau_L(g,x)}I_{\Phi^L_1} }\right] = \exp\left(-2 \pi \lambda \sum_{k=1}^2 b_k \int_x^\infty {\left(1-\frac{1}{\left(1+ \frac{a n C_L P_t G_k t^{- \alpha_L}}{\tau_L(g, x) N_L}\right)^{N_L}}\right) t p(t) dt }\right),\\
& \bbE_{\Phi^N_1}\left[e^{-\frac{a n}{\tau_L(g,x)}I_{\Phi^N_1} }\right] = \exp\left(-2 \pi \lambda \sum_{k=1}^2 b_k \int_{\psi_L(x)}^\infty {\left(1-\frac{1}{\left(1+ \frac{a n C_N P_t G_k t^{- \alpha_N}}{\tau_L(g, x) N_N}\right)^{N_N}}\right) t (1-p(t)) dt }\right),
\end{align}
\end{small}%
where $\psi_L(x)$ is the minimum distance for a NLOS interfering BS to have the same path-loss value of the LOS serving BS at a distance $x$ from the MS, $\psi_L(x)=\left(\frac{C_N}{C_L}\right)^{\frac{1}{\alpha_N}} x^{\frac{\alpha_L}{\alpha_N}}$ \cite{Bai2015}. The conditional NLOS  SINR coverage probability can be similarly bounded to have
\begin{small}
\begin{equation}
P^\text{TP}_\text{c$|$ $\mathsf{OBP}$,$\mathrm{N}$}\left(T\right) \leq \int_{0}^{\infty}{\int_{0}^{\infty}{ \mathbb{P}\left(I_{\Phi_1} < \tau_N\left(g,x\right) \right) f_{h_0^N} \left(g\right) } f_{r_0}^N(x) dg dx}. \label{eq:Cond3}
\end{equation}
\end{small}
where $\tau_N(g,x)=P_T G_\BS G_\MS g x^{-\alpha_N} \frac{C_N}{T}-\sigma^2$, and $f_{h_0^N} \left(g\right), f_{r_0}^N(x)$ are the probability distribution functions of the NLOS channel gain $h_0^N$ and the distance to the NLOS associated BS $r_0$, respectively. Further, the conditional interference probability can be written as
\begin{small}
\begin{equation}
\mathbb{P}\left(I_{\Phi+{1}} <\tau_L\left(g,x\right) \right) = \lim_{N \to \infty} \sum_{n=1}^{N} \dbinom{N}{n} \left(-1\right)^{n+1} \bbE_{\Phi^L_1}\left[e^{-\frac{a n}{\tau_N(g,x)}I_{\Phi^L_1} }\right]
\bbE_{\Phi^N_1}\left[e^{-\frac{a n}{\tau_N(g,x)}I_{\Phi^N_1}}\right],
\label{eq:Int2}
\end{equation}
\end{small}
with the interference Laplace transforms 
\begin{small}
\begin{align}
& \bbE_{\Phi^L_1}\left[e^{-\frac{a n}{\tau_N(g,x)}I_{\Phi^L_1} }\right] =\exp\left(-2 \pi \lambda \sum_{k=1}^2 b_k \int_{\psi_N(x)}^\infty {\left(1-\frac{1}{\left(1+ \frac{a n C_L P_t G_k t^{- \alpha_L}}{\tau_N(g, x) N_L}\right)^{N_L}}\right) t p(t) dt }\right),\\
& \bbE_{\Phi^N_1}\left[e^{-\frac{a n}{\tau_N(g,x)}I_{\Phi^N_1} }\right] = \exp\left(-2 \pi \lambda \sum_{k=1}^2 b_k \int_{x}^\infty {\left(1-\frac{1}{\left(1+ \frac{a n C_N P_t G_k t^{- \alpha_N}}{\tau_N(g, x) N_N}\right)^{N_N}}\right) t (1-p(t)) dt }\right).
\end{align} 
\end{small}
Finally, note that $\mathbb{P}\left(\mathsf{SBP}\right)=0$, which results from the zero MS side lobe assumption and the assumption that other BSs are sending data (i.e., using fixed beamforming vectors) while the considered BS and MS are training their beams in the full pilot reuse case. This ensures that the beam training procedure with the control SINR definition in \eqref{eq:c_SINR} and full pilot reuse will result in a joint BS-MS beamforming gain of either $G_\mathrm{BS} G_\mathrm{MS}$ or $0$, and this completes the proof. 
\end{proof}

\linespread{1.2}
\bibliographystyle{IEEEtran}

\end{document}

%% file: input.tex
\usepackage{amsfonts}
\usepackage{times}
\usepackage{graphicx}
\usepackage{latexsym}
\usepackage{dsfont}
\usepackage{amssymb}
\usepackage{amsmath}
\usepackage{cite}
\usepackage{verbatim}
\usepackage{subfigure}
\newtheorem{theorem}{Theorem}

\newenvironment{proof}{ \textbf{Proof:} }{ \hfill $\Box$}

\newcommand{\figref}[1]{{Fig.}~\ref{#1}}

\def\bb0{{\mathbb{0}}}

\def\ba{{\mathbf{a}}}
\def\bb{{\mathbf{b}}}

\def\bff{{\mathbf{f}}}

\def\bn{{\mathbf{n}}}

\def\bw{{\mathbf{w}}}

\def\b0{{\mathbf{0}}}

\def\bH{{\mathbf{H}}}
\def\bI{{\mathbf{I}}}

\def\bbE{{\mathbb{E}}}

\def\cR{\mathcal{R}}

\def\sf0{{\mathsf{0}}}

